\documentclass[number]{elsarticle}

\usepackage[utf8]{inputenc}
\usepackage[TS1,T1]{fontenc}
\usepackage{xspace}

\usepackage{amssymb, amsmath}
\usepackage[all]{xy}
\usepackage{xcolor}
\usepackage[english]{babel}
\usepackage{hyperref}
\usepackage{lmodern}

\usepackage{amsthm}
\newtheorem{theorem}{Theorem}[section]
\newtheorem{proposition}[theorem]{Proposition}
\newtheorem{lemma}[theorem]{Lemma}

\theoremstyle{definition}

\newtheorem{definition}[theorem]{Definition}

\theoremstyle{remark}

\newcommand\bbbone{{ \mathchoice {1\mskip-4mu\mathrm{l} } {1\mskip-4mu\mathrm{l} }{1\mskip-4.5mu\mathrm{l} } {1\mskip-5mu\mathrm{l}} }}

\makeatletter
\newcommand*{\remove}{%
  \mathpalette\@remove
}
\def\@remove#1#2{
\def\@removesymbol{\boldsymbol\backslash}%
  \mathord{%
    \rlap{%
      \settowidth\dimen@{$\m@th#1{#2}$}%
      \kern.5\dimen@
      \settowidth\dimen@{$\m@th#1{\@removesymbol}$}%
      \kern-.5\dimen@
      $\m@th#1{\@removesymbol}$%
    }%
    {#2}%
  }%
}
\makeatother

\newcommand\dd{\text{\textup{d}}} 
\newcommand\ds{\text{\textup{s}}} 
\newcommand\cdotaction{\mathord{\cdot}} 
\newcommand\exter{{\textstyle\bigwedge}} 
\newcommand\syme{{\textstyle\bigvee}} 
\newcommand\ordwedge{\mathord{\wedge}} 
\newcommand\ad{{\text{\textup{ad}}}} 
\newcommand\sign{{\text{\textup{sign}}}}
\newcommand\der{{\text{\textup{Der}}}} 
\newcommand\ensvide{{\varnothing}} 
\newcommand\Id{{\text{\textup{Id}}}} 
\newcommand\loc{{\text{\textup{loc}}}}
\newcommand\Ad{{\text{\textup{Ad}}}} 
\newcommand\lie{{\text{\textup{Lie}}}}
\newcommand\equ{{\text{\textup{equ}}}} 
\newcommand\Inn{{\text{\textup{Inn}}}} 
\newcommand\Out{{\text{\textup{Out}}}}

\newcommand{\Conn}{\Theta}				
\newcommand{\Act}{\mathcal{S}}
\newcommand{\dvol}{\textnormal{dvol}}

\newcommand{\hodgeast}{\mathop{\star}}
\newcommand{\dualast}{\ast}
\newcommand{\grast}{\bullet}

\newcommand{\rke}{\tau}
\newcommand{\tla}{{\lienotation{TLA}}}
\newcommand{\maxinner}{{\text{m.i.}}}
\newcommand{\inner}{{\text{\textup{inner}}}}
\newcommand{\algebraic}{{\text{\textup{alg.}}}}
\newcommand{\dR}{{\text{dR}}}

\newcommand{\circoverset}[1]{\mathring{#1}}
\newcommand{\omegadot}{\circoverset{\omega}}	
\newcommand{\Adot}{\circoverset{A}}	
\newcommand{\Rdot}{\circoverset{R}}	
\newcommand{\Fdot}{\circoverset{F}}	
\newcommand{\Conndot}{\circoverset{\Conn}}	
\newcommand{\nabladot}{\circoverset{\nabla}}

\newcommand{\halpha}{\widehat{\alpha}}
\newcommand{\homega}{\widehat{\omega}}	
\newcommand{\hgamma}{\widehat{\gamma}}	
\newcommand{\hphi}{\widehat{\phi}}	
\newcommand{\hg}{{\widehat{g}}}
\newcommand{\hnabla}{{\widehat{\nabla}}}

\newcommand{\hA}{{\widehat{A}}}				%
\newcommand{\hF}{{\widehat{F}}}				%
\newcommand{\hConn}{{\widehat{\Conn}}}
\newcommand{\hd}{\widehat{\dd}} 
\newcommand{\hR}{\widehat{R}}

\newcommand{\hsfX}{{\widehat{\mathsf{X}}}}

\newcommand{\raR}{{\widetilde{R}}}

\newcommand{\lfc}{\mathfrak{a}} 

\DeclareMathOperator{\Aut}{Aut} 
\DeclareMathOperator{\End}{End} 
\DeclareMathOperator{\tr}{tr} 

\newcommand\varnotation[1]{{\mathcal{#1}}}
\newcommand\algnotation[1]{{\mathbf{#1}}}
\newcommand\lienotation[1]{{\mathbf{\mathsf{#1}}}}
\newcommand\grnotation[1]{{\mathsf{#1}}}
\newcommand\evnotation[1]{{\mathnormal{#1}}}

\newcommand\varE{{\varnotation{E}}}

\newcommand\varL{{\varnotation{L}}}
\newcommand\varM{{\varnotation{M}}}
\newcommand\varP{{\varnotation{P}}}

\newcommand\algA{{\algnotation{A}}}

\newcommand\algzero{{\grnotation{0}}}

\newcommand\lieA{{\lienotation{A}}}

\newcommand\lieL{{\lienotation{L}}}

\newcommand\evF{{\evnotation{F}}}

\newcommand\kg{{\mathfrak g}}

\newcommand\kD{{\mathfrak D}}
\newcommand\kS{{\mathfrak S}} 

\newcommand\kX{{\mathfrak X}}
\newcommand\kY{{\mathfrak Y}}
\newcommand\ksl{{\mathfrak{sl}}}
\newcommand\ku{{\mathfrak{u}}}

\newcommand\gC{{\mathbb C}}

\newcommand\sfX{{\mathsf X}}

\newcommand\caG{{\mathcal G}}
\newcommand\caD{{\mathcal D}}

\newcommand\caL{{\mathcal L}}
\newcommand\caN{{\mathcal N}}
\newcommand\caZ{{\mathcal Z}}

\hypersetup{
pdfauthor={Cédric Fournel, Serge Lazzarini, Thierry Masson},
pdftitle={Gauge theories formulation on transitive Lie algebroids},
pdfsubject={Lie algebroids, connections, cohomology},
pdfcreator={pdflatex},
pdfproducer={pdflatex},
pdfkeywords={Lie algebroids}
}

\hypersetup{
plainpages=false,
colorlinks=true,
linkcolor=black, 
anchorcolor=black, 
citecolor=red, 
urlcolor=black, 
menucolor=black, 
filecolor=black, 
bookmarksopen=true,
bookmarksnumbered=true}

\numberwithin{equation}{section}

\usepackage[backgroundcolor=blue!20]{todonotes}

\begin{document}

\begin{frontmatter}

\title{Formulation of gauge theories on transitive Lie algebroids}

\author{Cédric Fournel}
\ead{cedric.fournel@cpt-univ-mrs.fr}

\author{Serge Lazzarini}
\ead{serge.lazzarini@cpt-univ-mrs.fr}

\author{Thierry Masson}
\ead{thierry.masson@cpt-univ-mrs.fr}

\address{Centre de Physique Théorique\tnoteref{umr}\\
Case postale 907, CNRS-Luminy\\
F--13288 Marseille Cedex 9, France}

\tnotetext[umr]{Unité Mixte de Recherche (UMR 7332) du CNRS, de l'université d'Aix-Marseille et de l'université du Sud Toulon-Var. Unité affiliée à la FRUMAM Fédération de Recherche 2291.}%


\begin{abstract}
In this paper we introduce and study some mathematical structures on top of transitive Lie algebroids in order to formulate gauge theories in terms of generalized connections and their curvature: metrics, Hodge star operator and integration along the algebraic part of the transitive Lie algebroid (its kernel). Explicit action functionals are given in terms of global objects and in terms of their local description as well. We investigate applications of these constructions to Atiyah Lie algebroids and to derivations on a vector bundle. The obtained gauge theories are discussed with respect to ordinary and to similar non-commutative gauge theories.
\end{abstract}

\end{frontmatter}

\newpage

\tableofcontents

\bigskip

Usually, gauge theories are mathematically understood within the framework of differential geometry in terms of connections on principal fibre bundles and covariant derivatives on associated vector bundles. Subsequently, non-commutative geometry has proposed a more general framework for gauge theories, in terms of associative algebras, differential calculi, modules and (non-commutative) connections. These (generalized) gauge theories have interesting features from a physical point of view (see \cite{Mass42} for a recent review), but they often require an investment in some more involved mathematical structures. 

In this paper, we propose another route to construct gauge theories which is grounded on transitive Lie algebroids. The first advantage of this approach is that the mathematical structures needed to write an action functional is very close to ordinary geometry, in particular for Atiyah Lie algebroids. The second benefit is that the natural gauge theories constructed in this paper will be shown to be of Yang-Mills-Higgs types, exactly as in many examples developed in non-commutative geometry.

This paper can be considered as a follow up of \cite{Mass38}, where some relations between non-commutative connections on a specific algebra and some generalized notions of connections on transitive Lie algebroids were exhibited. Inspired by \cite{Mass15}, these relations have been used as a guide for the present paper to construct the necessary mathematical structures on top of transitive Lie algebroids which permit to write gauge invariant action functionals. 

In Section~\ref{sec-localdescriptionforms}, we present some facts about the local descriptions of differential forms on a transitive Lie algebroid $\lieA$ with kernel $\lieL$. Our purpose is to establish a correspondence between differential complexes describing global objects on $\lieA$ and the corresponding differential complexes describing their local counterparts. This step is necessary in order to exhibit constructions proposed in the sequel, and in order to write an action functional in terms of local objects, as it is customary in physics.

In Section~\ref{sec-metric-integration}, we study metrics, Hodge star operators and integration on top of transitive Lie algebroids. These constructions take their roots from two previous constructions which used similar notions of forms and integration: the first one in the context of the derivation-based non-commutative geometry \cite{Mass15}, and the second one in the context of Lie algebroids \cite{Kuba96a,MR1908998,MR2020382}.

Integration on $\lieA$ have been studied in \cite{Kuba96a,MR1908998} for forms on $\lieA$ with values in functions. We generalize part of this work by defining an integration along the ``fibre'' (the algebraic or ``inner'' part of the transitive Lie algebroid) for forms with values in the kernel $\lieL$ (Def.~\ref{def-innerintegration}). In \cite{Kuba96a,MR1908998}, the geometric object which permits to define integration along the fibre is a non-singular cross-section $\varepsilon$ in $\exter^n \lieL$ where $n$ is the dimension of the fibre of the kernel $\lieL$. The starting point of our definition is quite different, and it is inspired by similar constructions proposed in \cite{Mass15}: we use right away a notion of metric on $\lieL$. Then we associate to such a metric a global form of maximal inner degree (Prop.~\ref{prop-globalinnerform}). This form plays the role of ``volume form'' for the integration along the fibre, and turns out to be dual to $\varepsilon$ in a certain sense, as will be explained. On the other hand, the notion of metric allows us to define a Hodge star operator as well. 

The definition and properties of metrics are given in \ref{subsec-metrics}, and the corresponding notions of integration and of Hodge star operator are defined in \ref{subsec-innerorientationandintegration} and \ref{subsec-hodgeoperators}. The notion of mixed local basis of forms introduced in \ref{subset-mixedlocalbasisofforms} plays an essential role in the set up of the definitions.

In Section~\ref{sec-gaugetheories}, using the mathematical structures introduced in Section~\ref{sec-metric-integration}, we write gauge invariant action functionals of Yang-Mills-Higgs type. These action functionals are given in terms of global objects on the Lie algebroid, the curvature of a generalized connection, and also in terms of trivialized forms on open subsets of the base manifold. The fact that the ``Higgs part'' of the generalized connection vanishes or not makes a clear distinction between pure Yang-Mills theories and Yang-Mills-Higgs type theories.

In Section~\ref{sec-applicationAtiyah} we apply the general constructions and results on specific Lie algebroids.

In \ref{subsec-applicationAtiyah}, we specify our constructions to Atiyah Lie algebroids for which the underlying geometry of the principal fibre bundle helps us to improve some of the results obtained in the general case. For instance, using integration along the algebraic part, Theorem~\ref{thm-atiyah-commutationofdifferentials} binds the differential on forms on $\lieA$ with values in functions to the de~Rham differential on forms on the base manifold, while Theorem~\ref{thm-relationsdeRhamTLAAtiyah} relates the de~Rham calculus on the principal fibre bundle to the space of forms with values in the kernel $\lieL$. Comments are made about the relations between ordinary gauge theories on a principal fibre bundle and our new gauge theories, as well as about common features and differences with the non-commutative geometry approach to gauge theories.

In \ref{subsec-Derivationsonavectorbundle}, we improve some general results for the case of Lie algebroids of derivations on a vector bundle. In that situation, we extend in \eqref{eq-definnerintegrationtrace} our notion of integration along the algebraic part, and we make apparent close relations with some non-commutative structures.

The more concrete physical applications of the gauge theories proposed here are out of the scope of the present paper. They will be exposed in a forthcoming paper.

\section{Local description of differential forms}
\label{sec-localdescriptionforms}

In this section, we recall some constructions of forms on Lie algebroids as well as their local descriptions.

\subsection{Differential forms on transitive Lie algebroids}

In this paper we use the notations introduced in \cite{Mass38}, and we refer to \cite{Mack05a} for more developments on Lie algebroids. Let $\varM$ be a smooth manifold. In this paper, $C^\infty(\varM)$ stands for complex valued functions on $\varM$ and $\Gamma(T\varM)$ is the space of smooth vector fields on $\varM$.

\begin{definition}
\label{def-liealgebroidalgebraic}
A Lie algebroid $\lieA$ is a finite projective module over $C^\infty(\varM)$ equipped with a Lie bracket $[-,-]$ and a $C^\infty(\varM)$-linear Lie morphism, the anchor, $\rho : \lieA \rightarrow \Gamma(T\varM)$ such that $[\kX, f \kY] = f [\kX, \kY] + (\rho(\kX)\cdotaction f) \kY$ for any $\kX, \kY \in \lieA$ and $f \in C^\infty(\varM)$, where $\Gamma(T\varM)$ is the space of smooth vector fields on $\varM$. 

A Lie algebroid $\lieA \xrightarrow{\rho} \Gamma(T\varM)$ is transitive if $\rho$ is surjective.
\end{definition}

For a transitive Lie algebroid, the kernel $\lieL = \ker \rho$ is a Lie algebroid (with null anchor) on $\varM$, and there is a locally trivial bundle in Lie algebras $\varL$ such that $\lieL = \Gamma(\varL)$. Such a Lie algebroid defines a short exact sequence of Lie algebras and $C^\infty(\varM)$-modules
\begin{equation}
\label{eq-sectransitiveliealgebroid}
\xymatrix@1{{\algzero} \ar[r] & {\lieL} \ar[r]^-{\iota} & {\lieA} \ar[r]^-{\rho} & {\Gamma(T\varM)} \ar[r] & {\algzero}}.
\end{equation}

The kernel $\lieL$ will often be referred to as the ``inner'' part of $\lieA$, and structures defined on $\lieL$ as ``inner'' objects. This terminology is inspired by the physical applications we have in mind, where $\Gamma(T\varM)$ will refer to (infinitesimal) symmetries on space-time (``outer'' symmetries) and $\lieL$ to (infinitesimal) inner symmetries \textsl{i.e.} infinitesimal gauge symmetries.

There are natural spaces of ``differential forms'' to be considered on a (transitive or not) Lie algebroid $\lieA$. They depend on the choice of a representation of $\lieA$ on a vector bundle. We are interested in two of them. 

\begin{definition}
\label{def-formsvaluesfunctions}
Let $\lieA \xrightarrow{\rho} \Gamma(T\varM)$ be a Lie algebroid (not necessarily transitive).
We define $(\Omega^\grast(\lieA), \hd_\lieA)$ as the graded commutative differential algebra of forms on $\lieA$ with values in $C^\infty(\varM)$. 
\end{definition}

\begin{definition}
\label{def-formsvalueskernel}
Let $\lieA \xrightarrow{\rho} \Gamma(T\varM)$ be a transitive Lie algebroid, with $\lieL$ its kernel.
We define $(\Omega^\grast(\lieA, \lieL), \hd)$ as the graded differential Lie algebra of forms on $\lieA$ with values in the kernel $\lieL$, where $\lieA$ is represented on $\lieL$ by the usual adjoint representation.
\end{definition}

We refer to \cite{Mass38} for properties of this differential calculus.

Recall that an ordinary connection on a transitive Lie algebroid $\lieA \xrightarrow{\rho} \Gamma(T \varM)$ is a splitting $\nabla : \Gamma(T \varM) \rightarrow \lieA$ as $C^\infty(\varM)$-modules of the short exact sequence
\begin{equation}
\label{eq-ordinaryconnectionontransitivealgebroid}
\xymatrix@1@C=25pt{{\algzero} \ar[r] & {\lieL} \ar[r]^-{\iota} & {\lieA} \ar[r]_-{\rho} & {\Gamma(T \varM)} \ar[r] \ar@/_0.7pc/[l]_-{\nabla} & {\algzero}}
\end{equation}
Then one can associate to $\nabla$ a $1$-form $\lfc^\nabla \in \Omega^1(\lieA, \lieL)$ uniquely defined by
\begin{equation*}
\kX = \nabla_X - \iota \circ \lfc^\nabla(\kX)
\end{equation*}
This $1$-form is normalized by $\lfc^\nabla \circ \iota(\ell) = -\ell$ for any $\ell \in \lieL$. In the following we will call it the connection $1$-form of the ordinary connection $\nabla$.

\medskip
Let us collect some of the definitions and notations introduced in \cite{Mass38}, Section~4.1. A trivial Lie algebroid is the Atiyah transitive Lie algebroid (see Section~\ref{sec-applicationAtiyah}) associated to the trivial principal fibre bundle $\varM \times G$ where $G$ is a Lie group, whose Lie algebra is denoted by $\kg$. The space of smooth sections $\tla(\varM, \kg) \equiv \lieA = \Gamma(T\varM \oplus (\varM \times \kg))$ is a transitive Lie algebroid with anchor $\rho(X \oplus \gamma) = X$, bracket $[X \oplus \gamma, Y \oplus \eta] = [X,Y] \oplus (X \cdotaction \eta - Y \cdotaction \gamma + [\gamma,\eta])$, and kernel $\lieL = \Gamma(\varM \times \kg)$. The graded commutative differential algebra $(\Omega^\grast(\lieA), \hd_\lieA)$ identifies with the total complex of the bigraded commutative algebra $\Omega^\grast(\varM) \otimes \exter^\grast \kg^\ast$ equipped with two differential operators $\dd$ and $\ds$ of bidegrees $(1,0)$ and $(0,1)$ respectively, where $\dd$ is the de~Rham differential on $\Omega^\grast(\varM)$, and $\ds$ is the Chevalley-Eilenberg differential on $\exter^\grast \kg^\ast$, so that $\hd_\lieA = \dd + \ds$. The graded differential Lie algebra $(\Omega^\grast_\tla(\varM,\kg), \hd_\tla) \equiv (\Omega^\grast(\lieA, \lieL), \hd)$ identifies with the total complex of the bigraded Lie algebra $\Omega^\grast(\varM) \otimes \exter^\grast \kg^\ast \otimes \kg$ equipped with the differential $\dd$ and the Chevalley-Eilenberg  differential $\ds'$ on $\exter^\grast \kg^\ast \otimes \kg$ for the adjoint representation of $\kg$ on itself, so that $\hd = \dd + \ds'$.

\subsection{Local trivializations}

As explained in detail in \cite{Mack05a}, a transitive Lie algebroid $\lieA \xrightarrow{\rho} \Gamma(T\varM)$ with kernel $\lieL = \Gamma(\varL)$ can be described locally as a triple $(U, \Psi, \nabla^0)$ where $\Psi : \Gamma(U \times \kg) \xrightarrow{\simeq} \lieL_U$ is an isomorphism of Lie algebras and $C^\infty(U)$-modules; where $\nabla^0 :  \Gamma(TU) \rightarrow \lieA_U$ is an injective morphism of Lie algebras and $C^\infty(U)$-modules compatible with the anchors; and such that $[\nabla^0_X, \iota \circ \Psi(\gamma) ] = \iota \circ \Psi(X \cdotaction \gamma)$ for any $X \in \Gamma(TU)$ and any $\gamma \in \Gamma(U \times \kg)$. Such a triple defines an isomorphism of Lie algebroids $S : \tla(U, \kg) \xrightarrow{\simeq} \lieA_U$ given by $S(X \oplus \gamma) = \nabla^0_X + \iota \circ \Psi(\gamma)$.

A Lie algebroid atlas for $\lieA$ is a family of triples $\{(U_i, \Psi_i, \nabla^{0,i})\}_{i \in I}$ such that $\bigcup_{i \in I} U_i = \varM$ and each triple $(U_i, \Psi_i, \nabla^{0,i})$ is a local trivialization of $\lieA$.

On $U_{ij} = U_i \cap U_j \neq \ensvide$ one can define $\alpha^{i}_{j} : U_{ij} \rightarrow \Aut(\kg)$ with $\alpha^{i}_{j} = \Psi_i^{-1} \circ \Psi_j$. To any $\kX \in \lieA$ there corresponds a family $\{ X^i \oplus \gamma^i \in \tla(U_i, \kg)\}_{i \in I}$ such that $S_i(X^i \oplus \gamma^i) = \kX_{|U_i}$. The local vector fields $X^i$ are the restrictions onto $U_i$ of the global vector field $X = \rho(\kX)$. There exists $\ell_{ij} \in \Omega^1(\lieA_{U_{ij}}, \lieL_{U_{ij}})$ such that $\nabla^{0,j}_X = \nabla^{0,i}_X + \iota \circ \ell_{ij}(X)$ and $\Psi_i(\gamma^i) = \Psi_j(\gamma^j) + \ell_{ij}(X)$. Anticipating on some Čech cohomology considerations, we make a distinction between $U_{ij}$ and $U_{ji}$ for $i \neq j$, so that one can define without any ambiguity $\chi_{ij} = \Psi_i^{-1} \circ \ell_{ij} \in \Omega^1(U_{ij}) \otimes \kg$. Then one has $\gamma^i = \alpha^{i}_{j}(\gamma^j) + \chi_{ij}(X)$. On $U_{ijk} = U_i \cap U_j \cap U_k \neq \ensvide$, one has the two cocycle relations $\alpha^{i}_{k} = \alpha^{i}_{j} \circ \alpha^{j}_{k}$ and $\chi_{ik} = \alpha^{i}_{j} \circ \chi_{jk} + \chi_{ij}$. The composite map $X \oplus \gamma^i \mapsto X \oplus \gamma^j \mapsto X \oplus \gamma^i$ on $U_{ij}$ gives $\alpha^{i}_{j} \circ \alpha^{j}_{i} = \Id \in \Aut(\kg)$ and $\alpha^{i}_{j} \circ \chi_{ji} + \chi_{ij} = 0$.
These expressions are compatible with the previous ones upon defining $\alpha^{i}_{i} = \Id \in \Aut(\kg)$ and $\chi_{ii} = 0$.

Using a local description of a transitive Lie algebroid, we can locally describe a form using the following definition.

\begin{definition}
\label{def-traivializationofforms}
Let $(U, \Psi, \nabla^0)$ be a local trivialization of $\lieA$. To any $q$-form $\omega \in \Omega^q(\lieA, \lieL)$ we define a local $q$-form $\omega_{\loc} \in \Omega^q_\tla(U,\kg)$ by
\begin{equation*}
\omega_{\loc} = \Psi^{-1} \circ \omega \circ S
\end{equation*}
\end{definition}

Given a Lie algebroid atlas for $\lieA$, one associates to $\omega \in \Omega^q(\lieA, \lieL)$ a family of local forms $\omega_{\loc}^i \in \Omega^q_\tla(U_i,\kg)$. For any $\kX_k \in \lieA$ with $1\leq k \leq q$, let $X_k \oplus \gamma^i_k \in \tla(U_i, \kg)$ be its family of trivializations. On any $U_{ij} = U_i \cap U_j \neq \ensvide$ one has
\begin{equation*}
\omega_{\loc}^i(X_1 \oplus \gamma^i_1, \ldots, X_q \oplus \gamma^i_q) = \alpha^{i}_{j} \circ \omega_{\loc}^j(X_1 \oplus \gamma^j_1, \ldots, X_q \oplus \gamma^j_q).
\end{equation*}
$s_{i}^{j} = S_j^{-1} \circ S_i : \tla(U_{ij}, \kg) \xrightarrow{\simeq} \tla(U_{ij}, \kg)$ is an isomorphism of (trivial) Lie algebroids and the previous relation takes the compact form
\begin{equation}
\label{eq-changelocaltrivializationforms}
\omega_{\loc}^i = \alpha^{i}_{j} \circ \omega_{\loc}^j \circ s_{i}^{j}.
\end{equation}
Let us define $\halpha_{j}^{\,i} : \Omega^q_\tla(U_{ij},\kg) \rightarrow \Omega^q_\tla(U_{ij},\kg)$ by
\begin{equation}
\label{eq-definitionalphahat}
\halpha_{j}^{\,i}(\omega_{\loc}^j) = \alpha^{i}_{j} \circ \omega_{\loc}^j \circ s_{i}^{j}.
\end{equation}

\begin{proposition}
\label{prop-relationtrivformslocaldifferentialcommutes}
A family of local forms $\{\omega_{\loc}^i\}_{i \in I}$ with $\omega_{\loc}^i \in \Omega^\grast_\tla(U_i,\kg)$ is a system of trivializations of a global form $\omega \in \Omega^\grast(\lieA, \lieL)$ if and only if 
\begin{equation}
\label{eq-relationtrivforms}
\halpha_{j}^{\,i}(\omega_{\loc}^j) = \omega_{\loc}^i
\end{equation}
for any $i,j$ such that $U_{ij} \neq \ensvide$.

For any $\omega \in \Omega^\grast(\lieA, \lieL)$, one has $\hd_\tla \omega_{\loc} = \Psi^{-1} \circ (\hd\omega) \circ S$ and, on $U_{ij} \neq \ensvide$, $\hd_\tla \omega_{\loc}^i = \halpha_{j}^{\,i} \big(\hd_\tla \omega_{\loc}^j \big)$.

The map $\halpha_{j}^{\,i} : \Omega^\grast_\tla(U_{ij},\kg) \rightarrow \Omega^\grast_\tla(U_{ij},\kg)$ is an isomorphism of graded differential Lie algebras.
\end{proposition}

\begin{proof}
\eqref{eq-relationtrivforms} is a direct consequence of the definition of the map $\halpha_{j}^{\,i}$.

$\hd_\tla \omega_{\loc} = \Psi^{-1} \circ (\hd\omega) \circ S$ and $\hd_\tla \omega_{\loc}^i = \halpha_{j}^{\,i} \big(\hd_\tla \omega_{\loc}^j \big)$ are straightforward computations using the definitions and the properties of $\hd$, $\hd_\tla$, $\Psi$, $\nabla^0$ and $S$. For the second relation, one needs the easy to establish relation
\begin{equation*}
X \cdotaction \alpha^{i}_{j}(\eta) + [ \gamma^i, \alpha^{i}_{j}(\eta)] = \alpha^{i}_{j}(X \cdotaction \eta + [\gamma^j, \eta]) 
\end{equation*}
for any $\eta \in \Gamma(U_{ij} \times \kg)$.
\end{proof}

We shall use the same notation $\halpha_{j}^{\,i} : \Omega^\grast_\tla(U_{ij}) \rightarrow \Omega^\grast_\tla(U_{ij})$ for the isomorphism defined by $\halpha_{j}^{\,i}(\omega^j_{\loc}) = \omega^j_{\loc} \circ s_{i}^{j}$, which permits to perform changes of trivializations for local expressions of forms in $\Omega^\grast(\lieA)$. Local descriptions of forms in $\Omega^\grast(\lieA)$ were presented in \cite{Kuba96a}.

\subsection{Mixed local basis of forms}
\label{subset-mixedlocalbasisofforms}

Let $\{\theta^a\}_{1 \leq a \leq n}$ be the dual basis of the basis $\{E_a\}_{1 \leq a \leq n}$ of $\kg$. Let $(U, \Psi, \nabla^{0})$ be a local trivialization of $\lieA$. Let $\nabla$ be an ordinary connection on $\lieA$. Then its connection $1$-form $\lfc$ has a local expression $\lfc_\loc = (A^a - \theta^a) E_a$ (summation over $a$ is understood), where $A \in \Omega^1(U)\otimes \kg$ is defined by $\lfc_{\loc}(X \oplus \gamma) = A(X) - \gamma$. Let us introduce the notation 
\begin{equation*}
\lfc^a = A^a - \theta^a \in \Omega^1_\tla(U).
\end{equation*}

\begin{definition}
\label{def-mixedbasis}
The local $1$-forms $\lfc^a$ on $U$ are called the mixed basis on the inner part of $\Omega^1_\tla(U)$ relative to the ordinary connection $\nabla$ and to the basis $\{E_a\}_{1 \leq a \leq n}$ of $\kg$.
\end{definition}

Let $\omega \in \Omega^p(\lieA, \lieL)$ and denote by $\omega_{\loc} \in \Omega^p_\tla(U,\kg)$ its trivialization over $U$. Then one has
\begin{equation*}
\omega_{\loc} = \sum_{r+s=p} \omega^{\theta}_{\mu_1 \ldots \mu_r a_1 \ldots a_s} \dd x^{\mu_1} \ordwedge \cdots \ordwedge \dd x^{\mu_r} \ordwedge \theta^{a_1} \ordwedge \cdots \ordwedge \theta^{a_s}
\end{equation*}
with $\omega^{\theta}_{\mu_1 \ldots \mu_r a_1 \ldots a_s} : U \rightarrow \kg$.
Using $\theta^a = A^a - \lfc^a$, this expression can be written as
\begin{equation}
\label{eq-omegalocinmixedbasis}
\omega_{\loc} = \sum_{r+s=p} \omega_{\mu_1 \ldots \mu_r a_1 \ldots a_s} \dd x^{\mu_1} \ordwedge \cdots \ordwedge \dd x^{\mu_r} \ordwedge \lfc^{a_1} \ordwedge \cdots \ordwedge \lfc^{a_s}
\end{equation}
for some new components $\omega_{\mu_1 \ldots \mu_r a_1 \ldots a_s} : U \rightarrow \kg$ which are polynomials in the $A^a_\mu$'s for $A^a = A^a_\mu \dd x^{\mu}$.

\begin{proposition}
\label{prop-changetrivialisationmixedbasis}
Define the matrix valued functions $G^{i}_{j}= \big({G^{i}_{j}}^b_a \big)_{1 \leq a,b \leq n}$ on $U_{ij} \neq \ensvide$ by $\alpha^{i}_{j}(E_a) = {G^{i}_{j}}^b_a E_b$ (summation over $b$). With obvious notations, on $U_{ij}$ one has 
\begin{equation*}
\lfc_i^a = {G^{i}_{j}}^a_b \lfc_j^b \circ s_{i}^{j}
\end{equation*}
where $s_{i}^{j} = S_j^{-1} \circ S_i : \tla(U_{ij}, \kg) \xrightarrow{\simeq} \tla(U_{ij}, \kg)$.
\end{proposition}

\begin{proof}
This is a direct consequence of \eqref{eq-changelocaltrivializationforms}.
\end{proof}

This relation can also be written as 
\begin{equation}
\label{eq-relationtriva}
\halpha_{i}^{\,j} (\lfc_i^a) = {G^{i}_{j}}^a_b \lfc_j^b
\end{equation}
where $\halpha_{i}^{\,j}$ is the one defined on forms in $\Omega^\grast_\tla(U_{ij})$. 

On $U_{ij} \neq \ensvide$, we can decompose $\omega^i_{\loc}$ along the $\lfc_i^a$'s and $\omega^j_{\loc}$ along the $\lfc_j^a$'s. Using \eqref{eq-relationtriva}, $\halpha_{j}^{\,i}(\omega^j_{\loc}) = \omega^i_{\loc}$, $\halpha_{j}^{\,i}(\omega^j_{\mu_1 \ldots \mu_r a_1 \ldots a_s}) = \alpha_{j}^{i}(\omega^j_{\mu_1 \ldots \mu_r a_1 \ldots a_s})$ and $\halpha_{j}^{\,i}(\dd x^\mu) = \dd x^\mu$, one gets
\begin{equation}
\label{eq-relationtrivcomponentsona}
\omega^i_{\mu_1 \ldots \mu_r a_1 \ldots a_s} = {G^{j}_{i}}^{b_1}_{a_1} \cdots {G^{j}_{i}}^{b_s}_{a_s} \alpha_{j}^{i}(\omega^j_{\mu_1 \ldots \mu_r b_1 \ldots b_s})
\end{equation}
These homogeneous gluing relations motivates the decomposition of global forms on $\lieA$ along the $\lfc^a$ instead of the $\theta^a$'s.

\section{Metrics and integration}
\label{sec-metric-integration}

In this section, we define a notion of metrics on transitive Lie algebroids and a notion of integration along the ``inner structure'' $\lieL$.

\subsection{Inner metrics}
\label{subsec-innermetrics}

Because $\varL$ is an ordinary vector bundle, the notion of metric on $\varL$ is well defined. In the following, a metric $h$ on $\varL$ will be referred to as an ``inner metric'' on $\lieL$. Such an inner metric is a $C^\infty(\varM)$-linear map $h : \lieL \otimes_{C^\infty(\varM)} \lieL \rightarrow C^\infty(\varM)$. 

In a local trivialization $(U, \Psi, \nabla^{0})$ of $\lieA$, $h$ is trivialized as a local map $h_\loc \in C^\infty(U) \otimes \syme^2 \kg^\dualast$ defined by $h_\loc(\gamma, \eta) = h(\Psi(\gamma), \Psi(\eta))$ for any $\gamma, \eta \in \Gamma(U \times \kg)$, where $\syme^\grast \kg^\dualast$ is the symmetric algebra over $\kg^\dualast$. Let $\{E_a\}_{1 \leq a \leq n}$ be a fixed basis of $\kg$. Then one can introduce the local components of $h$ over $U$ as $h_{ab} = h_\loc(E_a, E_b)$.

We can extend the inner metric $h$ to a $C^\infty(\varM)$-linear map 
\begin{equation*}
h : \Omega^p(\lieA, \lieL) \otimes_{C^\infty(\varM)} \Omega^q(\lieA, \lieL) \rightarrow \Omega^{p+q}(\lieA)
\end{equation*}
by
\begin{multline*}
h(\omega, \eta)(\kX_1, \dots, \kX_{p+q}) =\\
 \frac{1}{p!q!} \sum_{\sigma\in \kS_{p+q}} (-1)^{\sign(\sigma)} h(\omega(\kX_{\sigma(1)}, \dots, \kX_{\sigma(p)}), \eta(\kX_{\sigma(p+1)}, \dots, \kX_{\sigma(p+q)}))
\end{multline*}
for any $\omega \in \Omega^p(\lieA, \lieL)$ and $\eta \in \Omega^q(\lieA, \lieL)$. For $p=q=0$, this is the original map $h$. Notice that $h(\omega, \eta) = (-1)^{pq} h(\eta, \omega)$ and $h(\eta\omega_1, \omega_2) = \eta h(\omega_1, \omega_2)$ for any $\omega_1, \omega_2 \in \Omega^\grast(\lieA, \lieL)$ and $\eta \in \Omega^\grast(\lieA)$.

\begin{definition}
A Killing inner metric is an inner metric $h$ such that 
\begin{equation*}
h([\xi, \gamma], \eta) + h(\gamma, [\xi, \eta]) = 0
\end{equation*}
for any $\gamma, \eta, \xi \in \lieL$.

A locally constant inner metric is an inner metric $h$ such that the local components $h_{ab}$ of $h$ are constant functions in any local trivializations of $\lieA$.
\end{definition}

\begin{lemma}
\label{lem-killingwithform}
Let $h$ be a Killing metric on $\lieL$, then
\begin{equation*}
h([\eta,\omega_1],\omega_2) + (-1)^{q p_1} h(\omega_1,[\eta,\omega_2]) = 0
\end{equation*}
for any $\eta\in\Omega^{q}(\lieA,\lieL)$, $\omega_1\in\Omega^{p_1}(\lieA,\lieL)$ and $\omega_2\in\Omega^{p_2}(\lieA,\lieL)$.

Let $h$ be a locally constant Killing inner metric on $\lieL$. Then one has 
\begin{equation*}
\hd_\lieA h(\omega,\eta) = h(\hd\omega,\eta) + (-1)^p h(\omega,\hd\eta)
\end{equation*}
for any $\omega \in \Omega^{p}(\lieA,\lieL)$ and any $\eta \in \Omega^q(\lieA, \lieL)$.
\end{lemma}

\begin{proof}
These relations can be established in any local trivialization of $\lieA$.
\end{proof}

\subsection{Metrics}
\label{subsec-metrics}

\begin{definition}
Let $\lieA$ be a Lie algebroid over the manifold $\varM$. A metric on $\lieA$ is a symmetric, $C^\infty(\varM)$-linear map 
\begin{equation*}
\hg : \lieA \otimes_{C^\infty(\varM)} \lieA \rightarrow C^\infty(\varM)
\end{equation*}
\end{definition}

We do not suppose for the moment that $\hg$ is non-degenerate. This point will be discussed down below. This definition does not require $\lieA$ to be transitive, although we will only consider the transitive case in the following.

\begin{proposition}
\label{prop-constructionsaroundmetrics}
Let $\lieA \xrightarrow{\rho} \Gamma(T\varM)$ be a transitive Lie algebroid with kernel $\lieL$. 

A metric $\hg$ on $\lieA$ defines an inner metric $h = \iota^\ast \hg$ on $\lieL$. Explicitly one has
\begin{equation*}
h( \gamma, \eta) = \hg(\iota(\gamma), \iota(\eta))
\end{equation*}
for any $\gamma, \eta \in \lieL$. $h$ will be called the inner part of $\hg$.

Let $g$ be an ordinary metric on the manifold $\varM$. Then $\hg = \rho^\ast g$ is a metric on $\lieA$: \begin{equation*}
\hg(\kX, \kY) = g (\rho(\kX), \rho(\kY))
\end{equation*}
for any $\kX, \kY \in \lieA$.

Let $h$ be an inner metric on $\lieL$, and let $\nabla$ be an ordinary connection on $\lieA$. Denote by $\lfc \in \Omega^1(\lieA, \lieL)$ its associated connection $1$-form. Then the pair $(h, \nabla)$ defines a metric $\hg = {\lfc}^\ast h$ on $\lieA$ by
\begin{equation*}
\hg(\kX, \kY) = h (\lfc(\kX), \lfc(\kY))
\end{equation*}
This metric satisfies $h = \iota^\ast \hg$.
\end{proposition}

\begin{proof}
These claims are just direct applications of the definitions and the properties of the objects involved in the relations.
\end{proof}

The metric $\hg = \rho^\ast g$ vanishes on $\iota(\lieL)$ and the metric $\hg = {\lfc}^\ast h$ vanishes on the image of $\nabla$ in $\lieA$. In the following we will introduce a kind of notion of non degeneracy in order to get rid of such metrics.

\begin{definition}
A metric $\hg$ on $\lieA$ is inner non degenerate if its inner metric $h = \iota^\ast \hg$ is non degenerate on $\lieL$, \textit{i.e.} if it is non degenerate as a metric on $\varL$.
\end{definition}

The constructions given in Prop.~\ref{prop-constructionsaroundmetrics} help us to decompose any metric on $\lieA$ into ``smaller'' entities. 

\begin{proposition}
\label{prop-connectionassociatedtoinnernondegeneratemetric}
Let $\hg$ be an inner non degenerate metric on $\lieA$. Then there exists a unique connection $\nabla^{\hg}$ on $\lieA$ such that, for any $X \in \Gamma(T\varM)$ and any $\gamma \in \lieL$, 
\begin{equation}
\label{eq-blockdiagmetricconnection}
\hg(\nabla^\hg_X, \iota(\gamma)) = 0
\end{equation}
\end{proposition}

\begin{proof}
By a straightforward adaptation of the theorem of Riesz, the non degeneracy of $h$ implies that for any $C^\infty(\varM)$-linear map $\varpi : \lieL \rightarrow C^\infty(\varM)$ there exists a unique $\lfc \in \lieL$ such that $h(\lfc, \gamma) = \varpi(\gamma)$ for any $\gamma \in \lieL$. 

For any $\kX \in \lieA$, applying this result to $\varpi(\gamma) = - \hg(\kX, \iota(\gamma))$, there exists a unique $\lfc^\hg(\kX) \in \lieL$ such that $h(\lfc^\hg(\kX), \gamma) = - \hg(\kX, \iota(\gamma))$. By uniqueness, the map $\kX \mapsto \lfc^\hg(\kX) \in \lieL$ is $C^\infty(\varM)$-linear and one has $\lfc^\hg(\iota(\eta)) = -\eta$ for any $\eta \in \lieL$, so that $\lfc^\hg$ is a normalized $1$-form on $\lieA$ with values in $\lieL$. This implies that there exists a unique connection $\nabla^\hg : \Gamma(T \varM) \rightarrow \lieA$ with $\nabla^\hg_X = \kX + \iota \circ \lfc^\hg(\kX)$ for any $\kX \in \lieA$ with $X = \rho(\kX)$. 

By construction, one has $\hg(\nabla^\hg_X, \iota(\gamma)) = 0$ for any $X \in \Gamma(T \varM)$ and any $\gamma \in \lieL$.
\end{proof}

\begin{proposition}
\label{prop-tripleforinnernondegeneratemetric}
An inner non degenerate metric $\hg$ on $\lieA$ is equivalent to a triple $(g, h, \nabla)$ where $g$ is a  (possibly degenerate) metric on $\varM$, $h$ is a non degenerate inner metric on $\lieL$ and $\nabla$ is an ordinary connection on $\lieA$. 
The metric $\hg$ and the triple $(g, h, \nabla)$ are related by:
\begin{equation}
\label{eq-hatggha}
\hg(\kX, \kY) = g(\rho(\kX), \rho(\kY)) + h( \lfc(\kX), \lfc(\kY))
\end{equation}
where $\lfc$ is the connection $1$-form associated to $\nabla$.
\end{proposition}

\begin{proof}
It is obvious that such a triple defines an inner non degenerate metric $\hg$ by the proposed relation.

In the opposite direction, Prop.~\ref{prop-connectionassociatedtoinnernondegeneratemetric} defines a unique connection $\nabla$ associated to $\hg$ satisfying \eqref{eq-blockdiagmetricconnection}. Using $\kX = \nabla_X - \iota \circ \lfc(\kX)$ with $X = \rho(\kX)$ and $\lfc$ the connection $1$-form associated to $\nabla$, one has
\begin{equation*}
\hg(\kX, \kY) =
 \hg(\nabla_X, \nabla_Y) - \hg(\iota \circ \lfc(\kX), \nabla_Y) - \hg(\nabla_X, \iota \circ \lfc(\kY)) + \hg(\iota \circ \lfc(\kX), \iota \circ \lfc(\kY)).
\end{equation*}
The two terms in the middle vanish by construction of $\nabla$. Define now
\begin{align*}
g(X,Y) &= \hg(\nabla_X, \nabla_Y),
&
h(\gamma, \eta) &= \hg(\iota(\gamma), \iota(\eta)).
\end{align*}
The triple $(g, h, \nabla)$ satisfies the requirements. Notice that the inner metric $h$ in this construction is exactly $h = \iota^\ast \hg$.
\end{proof}

Let $(g,h,\nabla)$ be a triple as in Prop.~\ref{prop-tripleforinnernondegeneratemetric}. Let $(U, \Psi, \nabla^{0})$ be a local trivialization of $\lieA$. Denote by $\lfc^a \in \Omega^1_\tla(U)$ the mixed basis relative to $\nabla$ and $\{E_a\}_{1 \leq a \leq n}$. Let us assume that $U$ is the support of a chart of $\varM$, with coordinates $(x^\mu)$. 
Locally, we can write \eqref{eq-hatggha} as
\begin{equation*}
\hg_\loc = g_{\mu \nu} \dd x^\mu \dd x^\nu + h_{a b} \lfc^a \lfc^b
\end{equation*}
where $g_{\mu \nu}$ are the local components of the metric $g$. Thus, this mixed basis diagonalizes by blocks the local expression of the metric $\hg$.

\subsection{Inner orientation and integration}
\label{subsec-innerorientationandintegration}

Let as before $\{E_a\}_{1 \leq a \leq n}$ denote a basis of the $n$-dimensional Lie algebra $\kg$ and $\{\theta^a\}_{1 \leq a \leq n}$ its dual basis. Let $h$ be an inner metric on $\lieL$.

Let $G^{i}_{j}= \big({G^{i}_{j}}^b_a \big)_{1 \leq a,b \leq n}$ be defined as in Prop~\ref{prop-changetrivialisationmixedbasis}. The vector bundle $\varL$ is orientable if and only if $\det({G^{i}_{j}}) > 0$ for any $i,j$ such that $U_{ij} \neq \ensvide$. $\lieL$ is said to be orientable if $\varL$ is orientable.

\begin{definition}
A transitive Lie algebroid is inner orientable if its kernel is orientable.
\end{definition}

This notion of ``inner orientable Lie algebroid'' is the same as the notion of ``vertically orientable Lie algebroid'' used in \cite{MR1908998}.

On $U_i$, denote by $\gamma_i = \gamma_i^a E_a$ the local expression of an element $\gamma \in \lieL$. On $U_{ij} \neq \ensvide$, the relation $\gamma_i = \alpha^{i}_{j}(\gamma_j)$ induces the relation $\gamma_i^a = {G^{i}_{j}}^{a}_{b} \gamma_j^b$. For any $\gamma, \eta \in \lieL$, one has $h^j_\loc(\gamma_j, \eta_j) = h^i_\loc(\gamma_i, \eta_i) = h^i_\loc(\alpha^{i}_{j}(\gamma_j), \alpha^{i}_{j}(\gamma_j))$, so that
\begin{equation}
\label{eq-relationtrivh}
h^j_{{b_1} {b_2}} = {G^{i}_{j}}^{a_1}_{b_1} {G^{i}_{j}}^{a_2}_{b_2} h^i_{{a_1} {a_2}}
\end{equation}

\begin{proposition}
\label{prop-globalinnerform}
On each $U_i$, let $|h^i_\loc|$ denotes the absolute value of the determinant of the matrix $(h^i_{\loc})$. If $\lieL$ is orientable then on $U_{ij}\neq \ensvide$ one has
\begin{equation*}
\halpha_{i}^{\,j}\left( \sqrt{|h^i_\loc|}\; \lfc_i^{1} \ordwedge \cdots \ordwedge \lfc_i^{n}\right) = \sqrt{|h^j_\loc|}\; \lfc_j^{1} \ordwedge \cdots \ordwedge \lfc_j^{n}
\end{equation*}
This implies that there exists a global form $\omega_{h,\lfc} \in \Omega^\grast(\lieA)$ of maximal inner degree $n$ defined locally by
\begin{equation*}
\omega_{h,\lfc} = (-1)^n \sqrt{|h_\loc|}\; \lfc^{1} \ordwedge \cdots \ordwedge \lfc^{n}
\end{equation*}
\end{proposition}

The form $\omega_{h,\lfc} \in \Omega^\grast(\lieA)$ plays the role of a ``volume form'' for fibre integration.

\begin{proof}
On one hand, we have
\begin{align*}
\halpha_{i}^{\,j}\left( \lfc_i^{1} \ordwedge \cdots \ordwedge \lfc_i^{n} \right) 
&= {G^{i}_{j}}^1_{b_1} \cdots {G^{i}_{j}}^n_{b_n}\; \lfc_j^{b_1} \ordwedge \cdots \ordwedge \lfc_j^{b_n}
\\
&= \sum_{b_i} \varepsilon^{b_1 \dots b_n} {G^{i}_{j}}^1_{b_1} \cdots {G^{i}_{j}}^n_{b_n}\; \lfc_j^{1} \ordwedge \cdots \ordwedge \lfc_j^{n}
\\
&= \det({G^{i}_{j}})\; \lfc_j^{1} \ordwedge \cdots \ordwedge \lfc_j^{n}
\end{align*}
where $\varepsilon^{b_1 \dots b_n}$ is the totally antisymmetric Levi-Civita symbol.

On the other hand, a straightforward computation gives 
\begin{equation}
\det(h^j_\loc) = \det({G^{i}_{j}})^2 \det(h^i_\loc) \label{eq-relationtrivdeth}
\end{equation}
so that $|h^i_\loc| = |\det({G^{i}_{j}})|^{-2} |h^j_\loc|$ as a density.

Since $\det({G^{i}_{j}}) > 0$, one has $\halpha_{i}^{\,j}\left( \sqrt{|h^i_\loc|}\; \lfc_i^{1} \ordwedge \cdots \ordwedge \lfc_i^{n} \right) = \sqrt{|h^j_\loc|}\; \lfc_j^{1} \ordwedge \cdots \ordwedge \lfc_j^{n}$.
\end{proof}

According to \eqref{eq-omegalocinmixedbasis}, any form $\omega \in \Omega^\grast(\lieA, \lieL)$ of maximal degree $n$ in the inner direction can be written locally on $U_i$ as 
\begin{equation*}
\omega^i_{\loc} = (-1)^n \omega^\maxinner_{\loc\, i} \sqrt{|h^i_\loc|}\; \lfc_i^{1} \ordwedge \cdots \ordwedge \lfc_i^{n} + \omega^R
= \omega^\maxinner_{\loc\, i}\, \omega_{h,\lfc} + \omega^R
\end{equation*}
where $\omega^R$ contains only terms of lower degrees in the $\lfc_i^a$'s, with $\omega^\maxinner_{\loc\, i} \in \Omega^\grast(U_i) \otimes \kg$ (``$\maxinner$'' stands for ``maximum inner''). Notice that the factor $\omega^\maxinner_{\loc\, i}$ is the factor of $\sqrt{|h^i_\loc|}\; \theta^1 \ordwedge \cdots \ordwedge \theta^n$ in $\omega^i_{\loc}$, and $\omega^R$ is a sum of terms with degree in the $\theta^a$'s less or equal to $n-1$.

Two such local expressions can be compared on intersecting trivializations. Applying $\halpha_{i}^{\,j}$ on $\omega^i_{\loc}$ and using Prop.~\ref{prop-globalinnerform}, one gets
\begin{equation*}
\alpha^{j}_{i}(\omega^\maxinner_{\loc\, i}) = \omega^\maxinner_{\loc\, j}
\end{equation*}
so that the forms $\omega^\maxinner_{\loc\, i}$ define a global form $\omega^\maxinner \in \Omega^{\grast-n}(\varM, \varL)$, the space of (de~Rham) forms on $\varM$ with values in the vector bundle in Lie algebras $\varL$.

\begin{definition}
\label{def-innerintegration}
On an inner orientable transitive Lie algebroid equipped with a metric, one defines the inner integration as the operation
\begin{align*}
\int_\inner : \Omega^\grast(\lieA,\lieL) &\rightarrow \Omega^{\grast-n}(\varM, \varL)
&
\omega &\mapsto \omega^\maxinner
\end{align*}
This inner integration is zero when applied to forms which do not contain terms of maximal inner degree $n$.

\end{definition}

Because the form $\omega^\maxinner$ defined to be the result of this inner integration is in fact the factor of $\sqrt{|h^i_\loc|}\; \theta^1 \ordwedge \cdots \ordwedge \theta^n$, this integration does not depend on the choice of the connection but only on the inner metric $h$.

The same construction yields an inner integration for $C^\infty(\varM)$-valued forms through
\begin{equation*}
\int_\inner : \Omega^\grast(\lieA) \rightarrow \Omega^{\grast-n}(\varM)
\end{equation*}
Notice that by construction $\int_\inner \omega_{h,\lfc} =1$ where $\omega_{h,\lfc} \in \Omega^\grast(\lieA)$ is the volume form defined in Prop.~\ref{prop-globalinnerform}.

The global form $\omega_{h,\lfc}$ plays a dual role to the non-singular cross section $\varepsilon \in \exter^n \lieL$ used in \cite{MR1908998} to define integration along the fibre on $\Omega^\grast(\lieA)$. Given a non degenerate inner metric $h$ on $\lieL$, one can define
\begin{equation}
\label{eq-epsilonkubarski}
\varepsilon_{\loc} = (-1)^n \sqrt{|h_\loc|}^{\;-1}\; E_1 \ordwedge \cdots \ordwedge E_n
\end{equation}
in any local trivialization $(U, \Psi, \nabla^{0})$ of $\lieA$. These local expressions define a global form $\varepsilon \in \Gamma(\exter^n \varL) = \exter^n \lieL$ which satisfies $i_\varepsilon \omega_{h,\lfc} = 1$ where the operation $i_\varepsilon$ on $\Omega^\grast(\lieA)$ is defined as in \cite{MR1908998} and corresponds there to the integral along the fibre on $\Omega^\grast(\lieA)$. This relates our constructions to the ones proposed by Kubarski. The present notion of inner integration is also a direct generalization for transitive Lie algebroids of the notion of ``non-commutative'' integration defined and studied in \cite{Mass15} (see also \cite{Mass30} for constructions related to the present situation).

In order to define a global integration on forms, we suppose from now on that the manifold $\varM$ is orientable.

\begin{definition}
A transitive Lie algebroid is orientable if it is inner orientable and if its base manifold is orientable.
\end{definition}

Then we can define as follows.
\begin{definition}
The integration on an orientable transitive Lie algebroid equipped with a metric is the composition of the inner integration on $\Omega^\grast(\lieA)$ with the integration of forms on the base manifold. For any $\omega \in \Omega^\grast(\lieA)$, this integration is denoted by
\begin{equation*}
\int_\lieA \omega = \int_\varM \int_\inner \omega \in \gC.
\end{equation*}
\end{definition}

Obviously, this definition makes sense only when the integral on $\varM$ converges, which is always the case when $\varM$ is compact or for compactly supported (relative to $\varM$) forms on $\lieA$.

This definition is the same as the one given in Def.~2.1 in \cite{MR1908998}. The present definition only considers forms in $\Omega^\grast(\lieA)$. It can be extended to $\Omega^\grast(\lieA, \lieL)$ in the case of the transitive Lie algebroid of derivations of a vector bundle using the ``extended'' inner integration $\int^{\tr}_\inner$ which will be defined in \eqref{eq-definnerintegrationtrace}.

This integral is non zero only if its contains a non zero term which is of maximal degree in both the inner direction and the spatial direction. In that particular case, the integral depends only on this term. This integration has several properties which have been described in \cite{MR1908998}.

\begin{definition}
\label{def-scalarproductformsfunctions}
Let $\lieA$ be an orientable transitive Lie algebroid equipped with a metric. For any $\omega \in \Omega^\grast(\lieA)$ and $\eta \in \Omega^\grast(\lieA)$, one defines their scalar product as
\begin{equation*}
\langle \omega, \eta \rangle = \int_\lieA \omega\, \eta \in \gC
\end{equation*}
\end{definition}

\begin{definition}
\label{def-scalarproductforms}
Let $\lieA$ be an orientable transitive Lie algebroid equipped with a metric. For any $\omega \in \Omega^\grast(\lieA, \lieL)$ and $\eta \in \Omega^\grast(\lieA, \lieL)$, one defines their scalar product as
\begin{equation*}
\langle \omega, \eta \rangle = \int_\lieA h(\omega, \eta) \in \gC
\end{equation*}
\end{definition}

\subsection{Hodge star operator}
\label{subsec-hodgeoperators}

In the following, we suppose that $\lieA$ is an orientable transitive Lie algebroid equipped with an inner non-degenerate metric $\hg = (g,h,\nabla)$ such that $g$ is also a non-degenerate metric on $\varM$.

Let $\omega \in \Omega^p(\lieA, \lieL)$ be written locally in a trivialization $(U, \Psi, \nabla^{0})$ of $\lieA$ as in \eqref{eq-omegalocinmixedbasis}, where the $\lfc^{a}$'s are chosen to be the components of the local expression of the connection $1$-form associated to $\nabla$.

Consider the form in $\Omega^{m+n-p}_\tla(U,\kg)$ defined by
\begin{multline}
\label{eq-hodgestarlocal}
\hodgeast \omega_{\loc} = \sum_{r+s=p} (-1)^{s(m-r)}\; \frac{1}{r! s!}\; \sqrt{|h_\loc|} \sqrt{|g|}\; \omega_{\mu_1 \ldots \mu_r a_1 \ldots a_s}\; \epsilon_{\nu_1 \ldots \nu_m}\; \epsilon_{b_1 \ldots b_n} \\
\times g^{\mu_1 \nu_1} \cdots g^{\mu_r \nu_r}\; h^{a_1 b_1} \cdots h^{a_s b_s}\; \dd x^{\nu_{r+1}} \ordwedge \cdots \ordwedge \dd x^{\nu_m} \ordwedge \lfc^{b_{s+1}} \ordwedge \cdots \ordwedge \lfc^{b_n}
\end{multline}
where $\epsilon_{\nu_1 \ldots \nu_m}$ and $\epsilon_{b_1 \ldots b_n}$ are the totally antisymmetric Levi-Civita symbols, and where $(g^{\mu \nu})$ and $(h^{a b})$ are the inverse matrices of $(g_{\mu \nu})$ and $(h_{a b})$ respectively.

Using \eqref{eq-relationtriva}, \eqref{eq-relationtrivcomponentsona}, \eqref{eq-relationtrivh}, \eqref{eq-relationtrivdeth} one can establish that $\halpha_{j}^{\,i}(\hodgeast \omega_{\loc}^j) = \hodgeast \omega_{\loc}^i$ so that, by Prop.~\ref{prop-relationtrivformslocaldifferentialcommutes}, $\hodgeast \omega \in \Omega^{m+n-p}(\lieA, \lieL)$ is well-defined. 

\begin{definition}
The map $\hodgeast : \Omega^p(\lieA, \lieL) \rightarrow \Omega^{m+n-p}(\lieA, \lieL)$ is the Hodge star operator on the orientable transitive Lie algebroid $\lieA$ associated to the metric $\hg$.
\end{definition}

\begin{proposition}
For any $\omega \in \Omega^p(\lieA, \lieL)$ one has
\begin{equation*}
\hodgeast \hodgeast \omega = (-1)^{(m+n-p)p} \omega
\end{equation*}
\end{proposition}

\begin{proof}
This is just direct a computation using the definition of $\hodgeast$ and some combinatorial properties of the Levi-Civita symbols.
\end{proof}

This Hodge star operator defines a natural scalar product on any $\Omega^p(\lieA, \lieL)$ by
\begin{equation*}
( \omega, \eta ) = \langle \omega, \hodgeast \eta  \rangle
\end{equation*}
for any $\omega, \eta  \in \Omega^p(\lieA, \lieL)$ where $\langle -, - \rangle$ is defined in Def.~\ref{def-scalarproductforms}.

\begin{proposition}
\label{prop-hodgestarcontraction}
For any $\omega, \eta  \in \Omega^p(\lieA, \lieL)$ written in a trivialization of $\lieA$ as in \eqref{eq-omegalocinmixedbasis} one has
\begin{multline*}
( \omega, \eta ) =  (-1)^n \int_\varM \sum_{r+s=p} (-1)^{s(m-r)}\; (m-r)!\; (n-s)!\; \times
\\
\times\omega^{a}_{\mu_1 \ldots \mu_r a_1 \ldots a_s}\; \eta_{a}^{\mu_1 \ldots \mu_r a_1 \ldots a_s} \; \sqrt{|g|} \dd x^{1} \ordwedge \cdots \ordwedge \dd x^{m} 
\end{multline*}
with
\begin{equation*}
\eta_{a}^{\mu_1 \ldots \mu_r a_1 \ldots a_s} = g^{\mu_1 \nu_1} \cdots g^{\mu_r \nu_r}\; h^{a_1 b_1} \cdots h^{a_s b_s} h_{ab}\; \eta^{b}_{\nu_1 \ldots \nu_r b_1 \ldots b_s}.
\end{equation*}
\end{proposition}

\begin{proof}
This is just a combinatorial straightforward computation.
\end{proof}

Notice that the Hodge star operator $\hodgeast$ is also well defined on $\Omega^\grast(\lieA)$ where it permits to introduce a scalar product $( \omega, \eta ) = \langle \omega, \hodgeast \eta  \rangle$ using Def.~\ref{def-scalarproductformsfunctions}. A similar relation as the one given in the previous proposition can be established.

The Hodge star operator defined locally by \eqref{eq-hodgestarlocal} can be defined by the same relation on any differential calculus $\Omega^\grast(\lieA, \varE)$ where $\varE$ is a representation of the Lie algebroid $\lieA$ (see Def.~3.1 in \cite{Mass38}). This will be used in \ref{subsec-matterfields}.

\section{Gauge theories}
\label{sec-gaugetheories}

In this section we formulate gauge theories on transitive Lie algebroids. We use the notion of connections introduced in \cite{Mass38}, as well as its associated notion of infinitesimal gauge action of $\lieL$.

Here we use the terminology ``connection on $\lieA$'' for the notion of ``generalized connection on $\lieA$'' introduced in  \cite[Def.~3.18]{Mass38}. We will refer to ``ordinary connections'' to mention connections which are splitting of the short exact sequence~\eqref{eq-ordinaryconnectionontransitivealgebroid}.

In the following, any mention of gauge invariance under an infinitesimal gauge transformation $\xi \in \lieL$ means vanishing of the term in $\xi$ (but not necessarily of terms of higher orders in $\xi$).

\subsection{\texorpdfstring{Decomposition of a connection $1$-form and its curvature}{Decomposition of a connection 1-form and its curvature}}

\begin{definition}
Let $\homega \in \Omega^1(\lieA,\lieL)$ be a connection on $\lieA$. We define the reduced kernel endomorphism $\rke\in \End(\varL) \simeq \lieL \otimes_{C^\infty(\varM)} \lieL^\dualast$ associated to $\homega$ by
\begin{equation}
\rke = \homega \circ \iota + \Id_{\lieL}.
\end{equation}
\end{definition}

The following facts are direct consequences of this definition, and of Prop.~3.9 and Def.~3.19 in \cite{Mass38}.

\begin{lemma}

$\rke$ vanishes if and only if $\homega$ is an ordinary connection. 

The infinitesimal action of $\lieL$ on $\rke$ is given by $\rke^\xi=\rke+[\rke,\xi]$ for any $\xi\in\lieL$.
\end{lemma}

One can look at the reduced kernel endomorphism as an obstruction for $\homega$ to be an ordinary connection, so that, in some forthcoming developments, assuming $\rke=0$ will mean that we consider an ordinary connection.

\begin{definition}
We denote by $R_\rke : \lieL \times \lieL \rightarrow \lieL$ the obstruction for $\rke\in\End(\lieL)$ to be a endomorphism of Lie algebras: 
\begin{equation}
R_\rke(\gamma,\eta) = [ \rke(\gamma) , \rke(\eta) ] - \rke([\gamma,\eta])
\end{equation}
for any $\gamma, \eta \in \lieL$. $R_\rke$ is called the algebraic curvature of $\rke$.
\end{definition}

Let us introduce a fixed reference ordinary connection on $\lieA$, defined by a normalized $1$-form $\omegadot\in \Omega^1(\lieA,\lieL)$ (\textsl{i.e.} $\omegadot \circ \iota = - \Id_{\lieL}$). $\omegadot$ will be called a background connection on $\lieA$.

\begin{theorem}
\label{inducedconnection}
Let $\omegadot$ be a background connection on $\lieA$.
For any connection $\homega \in \Omega^1(\lieA,\lieL)$ with associated reduced kernel endomorphism $\rke$,
\begin{equation}
\label{eq-decompositionconnectiononeform}
\omega = \homega + \rke(\omegadot)
\end{equation}
is an ordinary connection on $\lieA$. The induced infinitesimal gauge action of $\lieL$ is the one on ordinary connections.
\end{theorem}

$\omega$ will be called the ordinary connection induced by $\homega$ relatively to $\omegadot$. Notice that when $\homega$ is an ordinary connection, one has $\rke = 0$, so that $\omega = \homega$. The background connection $\omegadot$ is only relevant for connections which are not ordinary connections.

\begin{proof}
These are straightforward computations.
\end{proof}

In a local trivialization $(U,\Psi,\nabla^0)$ of $\lieA$, one writes $\homega_\loc = \hA - \theta + \rke_\loc$, $\omegadot_\loc = \Adot - \theta$ and $\omega_\loc = A - \theta$, for $\rke_\loc \in C^\infty(U) \otimes \End(\kg)$ and $\hA, \Adot, A \in \Omega^1(U) \otimes \kg$. They are related by the relation $A = \hA + \rke_\loc(\Adot)$.

An ordinary connection on $\algA$ is a map $\nabla : \Gamma(T\varM) \rightarrow \lieA$. For any connection $\homega$ on $\lieA$, we introduce the generalization of this map as follows.

\begin{proposition}
\label{genconn}
Let $\hConn : \lieA \rightarrow \lieA$ be defined, for any $\kX \in \lieA$, by
\begin{equation}
\hConn(\kX) = \kX + \iota \circ \homega(\kX).
\end{equation}
Then $\hConn$ is a $C^\infty(\varM)$-linear map on $\lieA$, the curvature $\hR \in \Omega^2(\lieA,\lieL)$ of $\homega$ is given by
\begin{equation}
\iota\circ \hR(\kX, \kY) = [\hConn(\kX),\hConn(\kY)] - \hConn([\kX,\kY])
\end{equation}
and the infinitesimal gauge action of $\lieL$ on $\hConn$ is given by $\hConn^\xi = \hConn + [\hConn,\xi]$.

$\homega\in\Omega^1(\lieA,\lieL)$ is an ordinary connection if and only if $\hConn \circ \iota=0$.
\end{proposition}

\begin{proof}
The first part of the Proposition is just straightforward computations.

Let $\homega$ be an ordinary connection, and let $\nabla$ be its connection as in \eqref{eq-ordinaryconnectionontransitivealgebroid}. Then by definition $\iota \circ \homega(\kX) = \nabla_{\rho(\kX)} - \kX$, so that $\hConn = \nabla \circ \rho$. 

Conversely, if $\hConn \circ \iota = 0$, then, for any $\gamma \in \lieL$, one has $0 = \hConn(\iota(\gamma)) = \iota(\gamma) + \iota \circ \homega(\iota(\gamma))$, so that $\homega(\iota(\gamma)) = -\gamma$, which implies that $\homega$ is an ordinary connection on $\lieA$.
\end{proof}

Notice that when $\homega$ is an ordinary connection on $\lieA$, one has $\hConn^2=\hConn$, and $\hConn$ is the projection onto the image of $\nabla$ in $\algA$.

Let $\homega$ be a connection on $\lieA$, $\omegadot$ be a background connection on $\lieA$, and $\omega$ be the ordinary connection induced by $\homega$ relatively to $\omegadot$. Denote by $\hConn, \Conndot, \Conn : \lieA \rightarrow \lieA$ the maps associated to $\homega$, $\omegadot$ and $\omega$ respectively, and denote by $\nabladot, \nabla : \Gamma(T\varM) \rightarrow \lieA$ the connections as in \eqref{eq-ordinaryconnectionontransitivealgebroid} associated to $\omegadot$ and $\omega$ respectively. Then one has
\begin{align*}
\hConn(\kX) &= \Conn(\kX)+\iota\circ\rke(\kX-\Conndot(\kX)) \\
&= \nabla\circ\rho(\kX)+\iota\circ\rke(\kX-\nabladot\circ\rho(\kX))
\end{align*}
In the first expression, we identify $\rke$ with its induced map $\iota(\lieL) \rightarrow \lieL$.

We denote by $\Rdot, R \in \Omega^2(\varM, \varL)$ the curvature $2$-forms of the ordinary connections $\omegadot$ and $\omega$. Define $\hF = R - \rke \circ \Rdot \in \Omega^2(\varM, \varL)$. Notice that $\rho^\ast \hF \in \Omega^2(\lieA, \lieL)$.

For any $X \in \Gamma(T\varM)$ and $\gamma \in \lieL$ define
\begin{equation*}
(\caD_X \rke)(\gamma) = [\nabla_X, \rke(\gamma)] - \rke([\nabladot_X, \gamma]),
\end{equation*}
then $\caD_X \rke \in \End(\varL)$ and we can look at $\caD \rke$ as an element in $\Omega^1(\varM, \End(\varL))$. A straightforward computation shows that $\caD_X\caD_Y \rke - \caD_Y\caD_X \rke - \caD_{[X,Y]} \rke = [R(X,Y), \rke] - \rke([\Rdot(X,Y), \Id_\lieL])$ for any $X,Y \in \Gamma(T\varM)$. 

Denote by $\rho^\ast \caD \rke \in \Omega^1(\lieA, \End(\varL))$ its pull-back, given explicitly by
\begin{equation*}
(\rho^\ast \caD \rke)(\gamma) = [\nabla \circ \rho, \rke(\gamma)] - \rke([\nabladot \circ \rho, \gamma])
\end{equation*}
and notice that $((\rho^\ast \caD \rke)\circ \omegadot)(\kX,\kY) = (\caD_{\rho(\kX)} \rke)(\omegadot(\kY)) - (\caD_{\rho(\kY)} \rke)(\omegadot(\kX))$ defines an element in $\Omega^2(\lieA, \lieL)$.

Finally, one has $\omegadot^\ast R_\rke \in \Omega^2(\lieA, \lieL)$. Given all these notations, a straightforward computation shows the following.

\begin{proposition}
\label{prop-decompositionglobalcurvature}
The curvature $\hR \in \Omega^2(\lieA,\lieL)$ of $\homega$ can be written as
\begin{equation}
\label{eq-decompositioncurvature}
\hR=
\rho^\ast \hF
- (\rho^\ast \caD \rke)\circ \omegadot
+ \omegadot^\ast R_\rke
\end{equation}

Under an infinitesimal gauge transformation, each of the $3$ terms of this decomposition of $\hR$ have homogeneous transformations.
\end{proposition}

When $\homega$ is an ordinary connection, one has $\rke = 0$, so that $\hF = R$, and then $\hR = R$ as expected by the previously mentioned fact that $\homega = \omega$.

\subsection{Gauge invariant functional}

We suppose now that $\lieA$ is an orientable transitive Lie algebroid equipped with a non degenerate and inner non-degenerate metric $\hg = (g,h,\nabladot)$ such that $g$ is also a non-degenerate metric on $\varM$ and $h$ is a Killing inner metric on $\lieL$.

\begin{proposition}
For any connection $\homega \in \Omega^1(\lieA,\lieL)$ with curvature $2$-form $\hR$, we define the action functional:
\begin{equation}
\label{eq-functionalactiongauge}
\Act_\text{Gauge}[\homega] = \int_\lieA h(\hR,\hodgeast \hR).
\end{equation}
Then $\Act[\homega]$ is invariant under infinitesimal gauge transformations in $\lieL$.
\end{proposition}

Denote by $\dvol = \sqrt{|g|}\,\dd x^1\ordwedge\ldots\ordwedge\dd x^m$, where $|g|$ is the determinant of $g$, the volume form on $\varM$. Then define the Lagrangian density $\varL_\text{Gauge}[\homega]$ by
\begin{equation*}
\varL_\text{Gauge}[\homega]\, \dvol = \int_\inner h(\hR,\hodgeast \hR) \in \Omega^m(\varM)
\end{equation*}
Then the action functional is given by
\begin{equation*}
\Act_\text{Gauge}[\homega] = \int_\varM \varL_\text{Gauge}[\homega]\, \dvol
\end{equation*}

\begin{proof}
An infinitesimal gauge transformation $\xi \in \lieL$ induces the transformation $\hR \mapsto \hR^\xi = \hR + [\xi,\hR]$ on the curvature. At first order in $\xi$, one gets:
\begin{align*}
\varL_\text{Gauge}[\homega^\xi]\, \dvol &= \int_\inner h(\hR^\xi,\hodgeast \hR^\xi) = \int_\inner h(\hR + [\xi,\hR],\hodgeast \hR + \hodgeast[\xi,\hR])
\\
&= \int_\inner h(\hR,\hodgeast \hR) + \int_\inner h([\xi,\hR],\hodgeast \hR) + h(\hR,\hodgeast[\xi,\hR])
\\
&= \varL_\text{Gauge}[\homega]\, \dvol
\end{align*}
where we use Lemma~\ref{lem-killingwithform} in the last step.
\end{proof}

As a background connection, we choose $\omegadot$ to be the connection $1$-form associated to the connection $\nabladot$ in the triple $\hg = (g,h,\nabladot)$. We denote by $\omega$ the ordinary connection induced by $\homega$ relatively to $\omegadot$.

In a local trivialization $(U,\Psi,\nabla^0)$ of $\lieA$, one has
\begin{equation*}
(\caD_X \rke)_\loc(\gamma) = (X \cdotaction \rke_{\loc})(\gamma) + [A(X), \rke_{\loc}(\gamma)] - \rke_{\loc}([\Adot(X), \gamma])
\end{equation*}
for any $X \in \Gamma(TU)$ and $\gamma \in C^\infty(U) \otimes \kg$.

Let us introduce the following notations:
\begin{align*}
R_{\loc} &= F_{\mu \nu}^a E_a \dd x^\mu \ordwedge \dd x^\nu 
&
\Rdot_{\loc} &= \Fdot_{\mu \nu}^a E_a \dd x^\mu \ordwedge \dd x^\nu 
\\
(R_{\rke})_{\loc} (E_a, E_b) &= W_{ab}^c E_c
&
\rke_{\loc}(E_a) &= \rke^b_a E_b
\end{align*}
$F_{\mu \nu}^a$ and $\Fdot_{\mu \nu}^a$ are the ordinary field strengths of the connections $A$ and $\Adot$ respectively, and $\hF_{\mu \nu}^a = F_{\mu \nu}^a - \rke^a_b \Fdot_{\mu \nu}^b$. A direct computation shows that
\begin{equation*}
W_{ab}^c = \rke_a^d\rke_b^e C_{de}^c - C_{ab}^d \rke_d^c
\end{equation*}
where the $C_{ab}^c$'s are the structure constants of $\kg$ in the basis $\{ E_a \}_{1 \leq a \leq n}$. 

 With $(\caD_{\partial_\mu} \rke)_\loc(E_a) = (\caD \rke)_{\mu, a}^b E_b$, one has
\begin{equation*}
(\caD \rke)_{\mu, a}^b = \partial_\mu\rke_a^b + A_\mu^c \rke_a^d C_{cd}^b - \Adot_\mu^d C_{da}^c\rke_c^b. 
\end{equation*}

With these notations, using Prop~\ref{prop-hodgestarcontraction}, the Lagrangian density can be written as
\begin{multline}
\label{eq-decompositionlocalactiongauge}
\varL_\text{Gauge}[A, \rke] = \tfrac{\lambda_1}{4} g^{\mu_1 \mu_2} g^{\nu_1 \nu_2} h_{a_1 a_2} \hF_{\mu_1 \nu_1}^{a_1} \hF_{\mu_2 \nu_2}^{a_2}
\\
+\tfrac{\lambda_2}{2} g^{\mu_1 \mu_2} h^{a_1 a_2} h_{b_1 b_2} (\caD \rke)_{\mu_1, a_1}^{b_1} (\caD \rke)_{\mu_2, a_2}^{b_2}
\\
+\tfrac{\lambda_3}{4} h^{a_1 a_2} h^{b_1 b_2} h_{c_1 c_2} W_{a_1 b_1}^{c_1} W_{a_2 b_2}^{c_2}
\end{multline}
where $\lambda_1$, $\lambda_2$, $\lambda_3$ are combinatorial coefficients:
$\lambda_1=(-1)^n(m-2)!\,n!$,
$\lambda_2=(-1)^n(-1)^{m-1}(m-1)!\,(n-1)!$ and
$\lambda_3=(-1)^n m!\,(n-2)!$.

\subsection{Minimal coupling to matter fields}
\label{subsec-matterfields}

As explained in \cite{Mass38}, a connection defines a covariant derivative on the space of sections $\Gamma(\varE)$ of a vector bundle $\varE \rightarrow \varM$ which supports a representation $\phi : \lieA \rightarrow \kD(\varE)$ of $\lieA$. 

Using similar notations as in \cite{Mass38}, we denote by $\varphi \mapsto  \varphi^\xi = \varphi - \phi_\lieL(\xi)\varphi$ an infinitesimal gauge transformation performed on $\varphi \in \Gamma(\varE)$ by $\xi \in \lieL$.

\begin{proposition}
Let $\homega \in \Omega^1(\lieA, \lieL)$ be a connection on $\lieA$. For any $\varphi \in \Gamma(\varE)$, the map $\varphi \mapsto \hnabla^\varE \varphi = \phi(\hConn)\varphi$ defines a covariant derivative $\hnabla^\varE : \Gamma(\varE) \rightarrow \Omega^1(\lieA, \varE)$ which can be decomposed, using \eqref{eq-decompositionconnectiononeform}, as
\begin{equation}
\label{eq-decompositioncovariantderivative}
\hnabla^\varE \varphi = \rho^\ast\phi(\nabla)\cdotaction \varphi - (\phi_\lieL (\rke) \varphi) \circ \omegadot
\end{equation}
Under infinitesimal gauge transformations, each term has homogeneous transformations.
\end{proposition}

This covariant derivative is the minimal coupling of the connection $\homega$ with matter fields in $\Gamma(\varE)$.

\begin{proof}
Using the ordinary connection $\omega$, one has $\kX = \nabla_{\rho(\kX)} - \iota \circ \omega(\kX)$, so that $\phi(\kX) = \rho^\ast \phi(\nabla) - \phi_\lieL \circ \omega(\kX)$. Inserting this relation in $\hnabla^\varE_\kX \varphi = \phi(\kX) \varphi + \phi_\lieL \circ \homega(\kX) \varphi$, and using \eqref{eq-decompositionconnectiononeform}, one gets the decomposition. 

Under an infinitesimal gauge transformation of $\xi \in \lieL$, $\rke$ transforms homogeneously, so that, under an infinitesimal gauge transformation: $(\phi_\lieL (\rke) \varphi) \circ \omegadot \mapsto (\phi_\lieL (\rke) \varphi) \circ \omegadot + ([\phi_\lieL (\rke), \phi_\lieL (\xi)] \varphi) \circ \omegadot - (\phi_\lieL (\rke) \phi_\lieL (\xi) \varphi) \circ \omegadot = (\phi_\lieL (\rke) \varphi) \circ \omegadot - \phi_\lieL (\xi) (\phi_\lieL (\rke) \varphi) \circ \omegadot$ at first order in $\xi$.

$\rho^\ast\phi(\nabla)\cdotaction \varphi$ is just the ordinary covariant derivative $\nabla^\varE \varphi$ induced by $\omega$ on $\varE$, so that it transforms homogeneously under a gauge transformation.
\end{proof}

\begin{definition}
An metric $h^\varE$ on the vector bundle $\varE$ is $\phi_\lieL$-compatible if 
\begin{equation*}
h^\varE(\phi_\lieL(\xi)\varphi_1, \varphi_2) + h^\varE(\varphi_1,\phi_\lieL(\xi)\varphi_2) = 0
\end{equation*} 
for any $\varphi_1,\,\varphi_2\in\Gamma(\varE)$ and any $\xi \in \lieL$.
\end{definition}

This definition generalizes the notion of Killing inner metric which corresponds to the particular case with $\varE = \varL$ and the adjoint representation of $\lieA$ on $\lieL$.

\begin{proposition}
Let $h^\varE$ be a $\phi_\lieL$-compatible metric on $\varE$. Then, for any $\varphi_1, \varphi_2 \in \Gamma(\varE)$, $h^\varE(\varphi_1, \varphi_2)$ is invariant under infinitesimal gauge transformations.
\end{proposition}

\begin{proof}
This is a straightforward computation.
\end{proof}

Using the natural Hodge star operator defined on $\Omega^\grast(\lieA, \varE)$ (see remark at the end of \ref{subsec-hodgeoperators}), we can define the action functional
\begin{equation*}
\Act_\text{Matter}[\varphi,\homega] = \int_\lieA h^\varE(\hnabla^\varE \varphi, \hodgeast \hnabla^\varE \varphi)
\end{equation*}
This action functional is gauge invariant under infinitesimal gauge transformations in $\lieL$.

In a local trivialization of $\lieA$, this action functional can be written as
\begin{multline}
\label{eq-functionalactionmatterdeveloped}
\Act_\text{Matter}[\varphi,\homega] =
(-1)^n \int_\varM (m-1)! n!\; h^\varE(\hnabla^\varE_\mu \varphi, \hnabla^{\varE\,\mu} \varphi)
\\
+ (n-1)! m!\; h^\varE(\hnabla^\varE_a \varphi, \hnabla^{\varE\,a} \varphi)
\end{multline}
where
\begin{align*}
\hnabla^\varE_\mu &= \partial_\mu + A^a_\mu \phi_\lieL(E_a)
&
\hnabla^\varE_a &= - \rke^b_a \phi_\lieL(E_b)
\end{align*}

The first term in this functional action represents the square of the covariant derivative of $\varphi$ along the ordinary connection $\omega$. The second term represents a quadratic coupling of $\varphi$ with the fields $\rke^b_a$.

\subsection{Comments of these gauge theories}
\label{subsec-physicalremarks}

Using the decompositions \eqref{eq-decompositioncurvature} and \eqref{eq-decompositioncovariantderivative}, one gets the following structure for the total action functional $\Act[\varphi,\homega] = \Act_\text{Gauge}[\homega] + \Act_\text{Matter}[\varphi,\homega]$ constructed in \eqref{eq-functionalactiongauge}. 
\begin{subequations}
\label{eq-decaction}
\begin{align}
\Act[\varphi,\homega] = 
&\mathbin{\phantom{+}} \langle \rho^\ast \hF, \hodgeast \rho^\ast \hF \rangle 
\label{eq-decactionYM}
\\[2pt]
&+ \langle (\rho^\ast \caD \rke)\circ \omegadot, \hodgeast (\rho^\ast \caD \rke)\circ \omegadot \rangle 
\label{eq-decactionCovDerRKE}
\\[2pt]
&+ \langle R_\rke \circ \omegadot, \hodgeast R_\rke \circ \omegadot \rangle
\label{eq-decactionPot}
\\[2pt]
&+ \langle \rho^\ast\phi(\nabla)\cdotaction \varphi, \hodgeast \rho^\ast\phi(\nabla)\cdotaction \varphi \rangle
\label{eq-decactionCovDerPhi}
\\[2pt]
&+ \langle (\phi_\lieL (\rke) \varphi) \circ \omegadot, \hodgeast (\phi_\lieL (\rke) \varphi) \circ \omegadot \rangle
\label{eq-decactionMassPhi}
\end{align}
\end{subequations}
These terms are written locally in \eqref{eq-decompositionlocalactiongauge} and \eqref{eq-functionalactionmatterdeveloped}.

The gauge theories obtained in \eqref{eq-decaction} are of Yang-Mills-Higgs type. Indeed, the fields in the ordinary connection $\omega$ are Yang-Mills fields, and the $\rke$'s fields behave as Higgs fields, in the following way. The term \eqref{eq-decactionPot} vanishes when $\rke$ is a Lie algebra morphism, which can occur for instance when $\rke = \Id_\lieL$. Then, reporting this solution into \eqref{eq-decactionCovDerRKE} induces a mass term for the $A$'s fields. In a similar way, \eqref{eq-decactionMassPhi} induces a mass term for the matter fields $\varphi$. 

As a consequence, the gauge theories we have constructed here manifest one of the greatest strengths of non-commutative geometry, which is to produce a large class of natural Yang-Mills-Higgs type theories (see \cite{Mass42} for a recent review). This feature is a direct consequence of the short exact sequence \eqref{eq-sectransitiveliealgebroid} which corresponds, in non-commutative geometry, to the short exact sequence of groups $\xymatrix@1@C=15pt{{1} \ar[r] & {\Inn(\algA)} \ar[r] & {\Aut(\algA)} \ar[r] & {\Out(\algA)} \ar[r] & {1}}$ associated to any associative algebra $\algA$, where $\Aut(\algA)$ is the group of automorphisms of $\algA$, $\Inn(\algA)$ is its normal subgroup of inner automorphisms, and $\Out(\algA)$ is the quotient group of outer automorphisms. The infinitesimal version of this short exact sequence of groups, which involves the corresponding Lie algebras of derivations (see for instance \cite[eq. (4.10)]{Mass38}), was the key ingredient to show in \cite{Mass38} that connections in non-commutative geometry and connections on transitive Lie algebroids are related in some specific situations.

The precise study of the physical content of the present gauge theories is out of the scope of this paper. This will be elaborated in a forthcoming paper.

\section{Applications to specific Lie algebroids}
\label{sec-applicationAtiyah}

\subsection{Atiyah Lie algebroids}
\label{subsec-applicationAtiyah}

Let $G$ be a connected Lie group, and let $\kg$ be its Lie algebra. Let $\varP \xrightarrow{\pi} \varM$ be a $G$-principal bundle over $\varM$. The (transitive) Atiyah Lie algebroid of $\varP$ is defined by the short exact sequence
\begin{equation*}
\xymatrix@1{{\algzero} \ar[r] & {\Gamma_G(\varP, \kg)} \ar[r]^-{\iota} & {\Gamma_G(T\varP)} \ar[r]^-{\pi_\ast} & {\Gamma(T \varM)} \ar[r] & {\algzero}}
\end{equation*}
where 
\begin{align*}
\Gamma_G(T\varP) &= \{ \sfX \in \Gamma(T\varP) \, / \, \raR_{g\,\ast}\sfX = \sfX \text{ for all } g \in G \}
\\
\Gamma_G(\varP, \kg) &= \{ v : P \rightarrow \kg \, / \, v(p \cdotaction g) = \Ad_{g^{-1}} v(p) \text{ for all } g \in G \}.
\end{align*}
Here $\raR_g(p) = p \cdotaction g$ denotes the right action of $G$ on $\varP$ and $\iota$ is given by $\iota(v)(p) = \left( \frac{d}{dt} p \cdotaction e^{-t v(p)} \right)_{|t=0}$.

In order to get compact notations, we denote by $(\Omega^\grast_\lie(\varP, \kg), \hd)$ the space of forms on this Lie algebroid $\Gamma_G(T\varP)$ with values in its kernel $\Gamma_G(\varP, \kg)$ and by $(\Omega^\grast_\lie(\varP), \hd_\lie)$ the space of forms with values in $C^\infty(\varM)$.

The local description of the Lie algebroid $\Gamma_G(T\varP)$ is obtained using local trivializations of the principal fibre bundle $\varP$. Here we complete the exposition given in \cite{Mass38} in order to make apparent some relations used in forthcoming computations.

Let $\{(U_i, \varphi_i)\}$ be a system of trivializations of $\varP$ where $\varphi_i : \varP_{|U_i} \xrightarrow{\simeq} U_i \times G$, and denote by $s_i : U_i \rightarrow \varP$, with $s_i(x) = \varphi_i^{-1}(x,e)$, the associated local sections. One has $s_j(x) = s_i(x) g_{ij}(x)$ on any $U_{ij} \neq \ensvide$ where $g_{ij} : U_{ij} \rightarrow G$ are the associated transition functions. The isomorphism $\Psi_i : \Gamma(U_i \times \kg) \xrightarrow{\simeq} \Gamma_G(\varP_{|U_i},\kg)$ is given by $\Psi_i^{-1}(v) = s_i^\ast v$ for any $v \in \Gamma_G(\varP, \kg)$. With $p=s_i(x)\cdotaction g$, one has $\Psi_i(\eta^i)(p) = \Ad_{g^{-1}} \eta^i(x)$ for any $\eta^i \in \Gamma(U_i \times \kg)$ and one has $\nabla^{0, i}_{X\,|p} = T_{s_i(x)} \raR_{g} T_x s_i X_{| x} \in T_p \varP$. Any $\sfX \in \Gamma_G(T\varP)$ is trivialized over $U_i$ as $X\oplus \gamma^i$ where $X = \pi_\ast(\sfX)$ and $\gamma^i : U_i \rightarrow \kg$ represents the vertical part of $\sfX$ on $\varP$. More concretely, one has $\sfX = \nabla^{0, i}_{X} - \hgamma^{i\,\varP} \in \Gamma_G(T\varP)$ with $\hgamma^i(p) = \Psi_i(\gamma^i)(p) = \Ad_{g^{-1}} \gamma^i(x)$ and $\hgamma^{i\,\varP} = -\iota (\hgamma^i)$. 

On $U_{ij} \neq \ensvide$, a straightforward computation shows that
\begin{equation}
\label{eq-recollementgammai}
\gamma^i = g_{ij} \gamma^j g_{ij}^{-1} + g_{ij} \dd g_{ij}^{-1}(X),
\end{equation}
which gives
\begin{align}
\label{eq-alphachiatiyah}
\alpha_{j}^{i}(\gamma) &= g_{ij} \gamma g_{ij}^{-1}
&
\chi_{ij}(X) &= g_{ij} \dd g_{ij}^{-1}(X).
\end{align}

We will use the following result obtained in \cite{Mass38}. The space $\kg_\equ = \{ \xi^\varP \oplus \xi \ / \ \xi \in \kg \}$ is a sub Lie algebra of $\tla(\varP, \kg)$, where $\xi^\varP \in \Gamma(T\varP)$ is the fundamental vector field associated to $\xi \in \kg$ for the right action of $G$ on $\varP$. $\kg_\equ$ defines a Cartan operation on the differential complex $(\Omega^\grast_\tla(\varP,\kg), \hd_\tla)$. Denote by $(\Omega^\grast_\tla(\varP,\kg)_{\kg_\equ}, \hd_\tla)$ the differential graded subcomplex of basic elements. 

\begin{proposition}[\cite{Mass38}]
\label{prop-identificationdifferentialcalculusAtiyah}
Let $G$ be a connected and simply connected Lie group. Then $(\Omega^\grast_\lie(\varP, \kg), \hd)$ and $(\Omega^\grast_\tla(\varP,\kg)_{\kg_\equ}, \hd_\tla)$ are isomorphic as differential graded complexes. The same is true for $\kg_\equ$-basic forms in $\Omega^\grast_\tla(\varP)$ and $\Omega^\grast_\lie(\varP)$.
\end{proposition}

From now on we suppose that $G$ is connected and simply connected, so that the identifications of Prop.~\ref{prop-identificationdifferentialcalculusAtiyah} apply.

We associate to a form $\omega \in \Omega^\grast_\lie(\varP, \kg)$ its family of local forms $\{\omega_{\loc}^i\}_{i \in I}$ with $\omega_{\loc}^i \in \Omega^\grast_\tla(U_i,\kg)$ satisfying \eqref{eq-relationtrivforms}. Let $\homega \in \Omega^\grast_\tla(\varP,\kg)_{\kg_\equ}$ be the $\kg_\equ$-basic form corresponding to $\omega \in \Omega^\grast_\lie(\varP, \kg)$ in the identification of Prop.~\ref{prop-identificationdifferentialcalculusAtiyah}.

\begin{lemma}
\label{lemma-pullbackbasicforms}
One has $\omega_{\loc}^i = s_i^\ast \homega \in \Omega^\grast_\tla(U_i,\kg)$.
\end{lemma}

\begin{proof}
Let us first recall some key features of the identification of the differential calculus $(\Omega^\grast_\lie(\varP, \kg), \hd)$ with $(\Omega^\grast_\tla(\varP,\kg)_{\kg_\equ}, \hd_\tla)$. In \cite{Mass38} a short exact sequence of Lie algebras and $C^\infty(\varM)$-modules
\begin{equation*}
\xymatrix@1{{\algzero} \ar[r] & {\caZ} \ar[r] & {\caN} \ar[r]^-{\rho_\varP} & {\Gamma_G(T\varP)} \ar[r] & {\algzero}}
\end{equation*}
is used where $\caZ$ is defined to be the $C^\infty(\varP)$-module generated by $\kg_\equ$ and $\caN = \Gamma_G(T\varP) \oplus \caZ \subset \tla(\varP, \kg)$. It is shown that $\caN$ generates the space $\tla(\varP, \kg)$ as a $C^\infty(\varP)$-module.
The isomorphism $\lambda : \Omega^\grast_\tla(\varP,\kg)_{\kg_\equ} \rightarrow \Omega^\grast_\lie(\varP, \kg)$ is explicitly defined as follows. For any $\homega \in \Omega^r_\tla(\varP,\kg)_{\kg_\equ}$, for any $\sfX_1, \dots, \sfX_r \in \Gamma_G(T\varP)$, denote by $\hsfX_1, \dots, \hsfX_r \in \caN$ any family such that $\rho_\varP(\hsfX_i) = \sfX_i$, then the map $p \mapsto \lambda(\homega)(\sfX_1, \dots, \sfX_r)(p) = \homega(\hsfX_1, \dots, \hsfX_r)(p) \in \kg$ is a $G$-equivariant map.

In order to simplify the exposition, we prove the lemma for $1$-forms. The algebraic machinery is the same for $p$-forms. For any $\sfX \in \Gamma_G(T\varP)$ and any $x \in U_i$, by definition one has
\begin{equation*}
\omega_{\loc}^i (X \oplus \gamma^i)(x) 
= \Psi_{i}^{-1} (\omega(\sfX))(x) 
= \omega(\sfX)(s_i(x))
= \homega_{s_i(x)}( \hsfX_{|s_i(x)})
\end{equation*}
for any $\hsfX \in \caN$ such that $\rho_\varP(\hsfX) = \sfX$. On $\varP_{|U_i}$, let us take $\hsfX = \nabla^{0, i}_{X} - \hgamma^{i\,\varP} + (\hgamma^{i\,\varP} \oplus \hgamma^{i}) = \nabla^{0, i}_{X} \oplus \hgamma^{i}$ where $\hgamma^{i\,\varP} \oplus \hgamma^{i} \in \caZ$ \cite{Mass38}. Then $\hsfX_{|s_i(x)} = (s_{i\, \ast} X)_{|x} \oplus \gamma^i(x)$ by construction of $\nabla^{0, i}_{X}$ and $\hgamma^i$. This gives $\omega_{\loc}^i (X \oplus \gamma^i) = (s_i^\ast \homega)(X \oplus \gamma^i)$.
\end{proof}

We can summarize the identifications between these differential calculi in the following diagram:
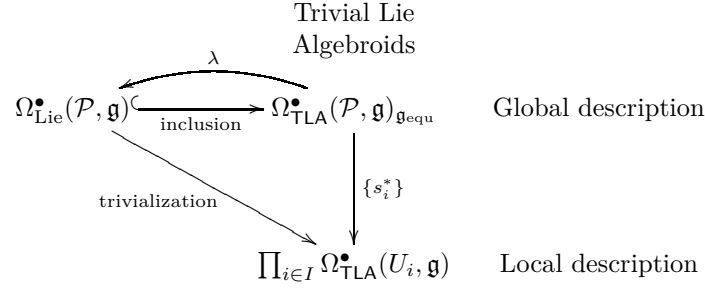
\begin{equation*}
\xymatrix@R=10pt@C=10pt{
 & & & {\text{\parbox{5em}{\centering Trivial Lie\\ Algebroids}}} & \\
{\Omega^\grast_\lie(\varP, \kg)} \ar@{^{(}->}[rrr]_-{\text{inclusion}} \ar[dddrrr]_(0.45)*!/^-3pt/{\text{\scriptsize trivialization}} 
& & & 
{\Omega^\grast_\tla(\varP,\kg)_{\kg_\equ}} \ar[ddd]^-{\{ s_i^\ast\}} \ar@/_1.2pc/[lll]_-{\lambda}
& {\text{Global description}} \\
& & & & \\
& & & & \\
& & & {\prod_{i\in I} \Omega^\grast_\tla(U_i,\kg)} & {\text{Local description}}
}
\end{equation*}

\medskip
From now on, we suppose that $\kg$ is semi-simple, so that its Killing form $k$ is non degenerate.

On a trivialization of $\Gamma_G(T\varP)$ associated to a trivialization $(U_i, \varphi_i)$ of $\varP$, we define $h^i_\loc(\gamma,\eta) = k(\gamma, \eta)$ for any $\gamma, \eta : U_i \rightarrow \kg$. Then, using the invariance of $k$ under the adjoint action of $G$ on $\kg$, we get that the $h^i_\loc$'s define a global metric $h$ on $\lieL$ (see \eqref{eq-relationtrivh} and \eqref{eq-alphachiatiyah}).

Let us introduce a fixed connection $\nabla$ on the Lie algebroid $\Gamma_G(T\varP)$, \textsl{i.e.} an ordinary connection on the principal fibre bundle $\varP$. The mixed basis on any trivialization of $\Gamma_G(T\varP)$ will be defined relative to this connection. 

$G$ being connected, the vector bundle $\caL = \varP \times_\Ad \kg$ is orientable, so that $\Gamma_G(T\varP)$ is inner orientable.

\begin{theorem}
\label{thm-atiyah-commutationofdifferentials}
Suppose that the inner metric $h$ is such that $\sqrt{|h_\loc|}$ is locally constant in any local trivialization of $\Gamma_G(T\varP)$, and that the Lie algebra $\kg$ is unimodular. Then for any $\omega \in \Omega^\grast_\lie(\varP)$ one has
\begin{equation*}
\int_\inner \hd_\lie \omega = \dd \int_\inner \omega
\end{equation*}
\end{theorem}

A Lie algebra is unimodular in our sense if the trace of its adjoint action vanishes. When the group is finite dimensional and connected, this definition of unimodularity is equivalent to the one defined on $G$ using Haar measures \cite{Bour72b}. There are well-known sufficient conditions for a group to be unimodular: compact, abelian, connected reductive or nilpotent etc.

This Theorem is similar to Theorem~1.2 in \cite{MR1908998} or Theorem~5.2.2 in \cite{Kuba96a}. A key condition required in these theorems is that the cross-section $\varepsilon \in \exter^n \lieL$ which defines the integral along the fibre be invariant oriented, which means that it is invariant under the $\exter^n \ad$ representation of $\lieA$ on $\exter^n \lieL$. Using the definition of $\varepsilon$ given by \eqref{eq-epsilonkubarski} in our context, this is equivalent to both $|h_\loc|$ being locally constant and the Lie algebra being unimodular. In the following proof, we will only use these two conditions. 

For instance, the inner metric $h$ defined above by the Killing metric $k$ is such that $\sqrt{|h_\loc|}$ is locally constant in any local trivialization of $\Gamma_G(T\varP)$.

\begin{proof}
Denote by $\ds$ the Chevalley-Eilenberg differential on $\exter^\grast \kg^\ast$, which satisfies $\ds \theta^c = - \frac{1}{2} C^{c}_{a b} \theta^a \ordwedge \theta^b$. Then one has
\begin{equation}
\label{eq-soninnerlessthatn}
\ds (\theta^{a_1} \ordwedge \cdots \ordwedge \theta^{a_{n-1}}) = (-1)^{n} \tr (C_{a_n})\; \theta^{a_1} \ordwedge \cdots \ordwedge \theta^{a_n}
\end{equation}
where $C_{a_n}$ is the matrix $(C_{a_n a}^{b})_{a,b}$ and $a_n$ is the missing index in the fixed multi-index $(a_1, \dots, a_{n-1})$ with $a_k \neq a_\ell$ for $k \neq \ell$.

Using this result, let us now collect the factor of $\sqrt{|h_\loc|}\; \lfc^{1} \ordwedge \cdots \ordwedge \lfc^{n}$ in $(\dd + \ds) \omega_{\loc}$ when one uses the decomposition 
\begin{equation*}
\omega_{\loc} = (-1)^n \omega^\maxinner_{\loc} \sqrt{|h_\loc|}\; \lfc^{1} \ordwedge \cdots \ordwedge \lfc^{n} + \omega^R
\end{equation*}
with $\omega^\maxinner_{\loc} \in \Omega^\grast(U)$ and $\omega^R$ containing only terms of degrees $< n$ in the $\lfc^a$'s.

$\dd \sqrt{|h_\loc|} \lfc^{1} \ordwedge \cdots \ordwedge \lfc^{n}$ and $\dd \omega^R$ do not contribute because $\sqrt{|h_\loc|}$ is locally constant and the degrees do not match for other terms. When $\kg$ is unimodular, \eqref{eq-soninnerlessthatn} implies that $\ds \sqrt{|h_\loc|} \lfc^{1} \ordwedge \cdots \ordwedge \lfc^{n}$ and $\ds \omega^R$ do not contribute. The remaining term is then $\dd \omega^\maxinner_{\loc}$, which is globally $\dd \int_\inner \omega$.
\end{proof}

From now on, we suppose that $G$ is connected, simply connected, semi-simple, unimodular and of dimension $n$. In other words, $G$ is the connected and simply connected group associated to a semi-simple unimodular $n$-dimensional Lie algebra $\kg$. As before, $k$ is the Killing metric on $\kg$.

\begin{lemma}
\label{lem-pullbackvolumeform}
The $\kg_\equ$-basic form in $\Omega^\grast_\tla(\varP)$ corresponding to the volume form $\omega_{h,\lfc} \in \Omega^\grast(\lieA)$ is
\begin{equation*}
\homega_{k,\nabla} = (-1)^n \sqrt{|k|}\; (\omega_\nabla^{1} - \theta^1) \ordwedge \cdots \ordwedge (\omega_\nabla^{n} - \theta^n)
\end{equation*}
where $\omega_\nabla = \omega_\nabla^{a} \otimes E_a \in \Omega^\grast(\varP) \otimes \kg$ is the (ordinary) connection $1$-form on $\varP$ associated to $\nabla$.
\end{lemma}

\begin{proof}
Notice that the proof of Lemma~\ref{lemma-pullbackbasicforms} applies to basic forms in $\Omega^\grast_\tla(\varP)$ which are trivialized as local forms in $\Omega^\grast_\tla(U_i) = \Omega^\grast(U_i) \otimes \exter^\grast \kg^\ast$ via the pull-back $s_i^\ast$. Using the unimodular property of $\kg$, a straightforward computation shows that the proposed expression for $\homega_{k,\nabla}$ is $\kg_\equ$-basic.

Now, $\omega_{h,\lfc}$ is locally defined on $U_i$ as $(-1)^n \sqrt{|h^i_\loc|}\; \lfc_i^{1} \ordwedge \cdots \ordwedge \lfc_i^{n}$ where $h^i_\loc = k$ and $\lfc_i^a = A_i^a - \theta^a \in \Omega^1_\tla(U_i)$. Notice finally that $A_i = A_i^a \otimes E_a$ is the local expression of $\omega_\nabla$ given explicitly by $A_i = s_i^\ast \omega_\nabla$.
\end{proof}

For any $\omega = \omega_{\dR} \otimes \omega_{\algebraic} \otimes \xi \in \Omega^\grast(\varP) \otimes \exter^\grast \kg^\ast \otimes \kg = \Omega^\grast_\tla(\varP,\kg)$, we now define a natural map $\int_{\algebraic} \Omega^\grast_\tla(\varP,\kg) \rightarrow \Omega^{\grast-n}(\varP) \otimes \kg$ by
\begin{equation*}
\int_{\algebraic} \omega  = 
\begin{cases}
\omega_{\dR} \otimes \xi & \text{ if $\omega = \omega_{\dR} \otimes \sqrt{|k|}\; \theta^1 \ordwedge \cdots \ordwedge \theta^n \otimes \xi$}\\
0 & \text{ if $\omega_{\algebraic} \not\in \exter^n \kg^\ast$}
\end{cases}
\end{equation*}

\begin{theorem}
\label{thm-relationsdeRhamTLAAtiyah}
The following diagram is commutative
\begin{equation*}
\xymatrix@R=8ex{
  {\Omega^\grast_\lie(\varP, \kg)} \ar@{^{(}->}[r] \ar[d]_-{\int_\inner}
& {\Omega^\grast_\tla(\varP,\kg)} \ar[d]^-{\int_{\algebraic}} 
\\
  {\Omega^{\grast-n}(\varM, \varL)} \ar@{^{(}->}[r] 
& {\Omega^{\grast-n}(\varP) \otimes \kg} 
 }
\end{equation*}
\end{theorem}

In this diagram, $\Omega^{\grast}(\varM, \varL)$ is identified with the space of tensorial forms in $\Omega^{\grast}(\varP) \otimes \kg$ \cite{KobaNomi96a}.

\begin{proof}
The first point to check is that $\int_{\algebraic}$ maps $\kg_\equ$-basic forms in $\Omega^\grast_\tla(\varP,\kg)$ to tensorial forms in $\Omega^{\grast-n}(\varP) \otimes \kg$. Because $G$ is connected and simply connected, a form $\sum_a \homega^a_{\dR} \otimes \xi_a \in \Omega^{\grast}(\varP) \otimes \kg$ is tensorial if and only if $ \sum_a (L_{\xi^\varP} \homega^a_{\dR}) \otimes \xi_a + \sum_a \homega^a_{\dR} \otimes [\xi, \xi_a] = 0$ and $ \sum_a (i_{\xi^\varP}\homega^a_{\dR}) \otimes \xi_a = 0$ for any $\xi \in \kg$.

A form $\homega = \sum_a \homega^a_{\dR} \otimes \sqrt{|k|}\; \theta^1 \ordwedge \cdots \ordwedge \theta^n \otimes \xi_a$ is basic if and only if for any $\xi \in \kg$ one has
\begin{multline*}
\sum_a (L_{\xi^\varP} \homega^a_{\dR}) \otimes \sqrt{|k|}\; \theta^1 \ordwedge \cdots \ordwedge \theta^n \otimes \xi_a 
\\
+ 
\sum_a \homega^a_{\dR} \otimes (L^\kg_\xi \sqrt{|k|}\; \theta^1 \ordwedge \cdots \ordwedge \theta^n )\otimes \xi_a
\\
+
\sum_a \homega^a_{\dR} \otimes \sqrt{|k|}\; \theta^1 \ordwedge \cdots \ordwedge \theta^n \otimes [\xi,\xi_a] = 0
\end{multline*}
and
\begin{multline*}
\sum_a (i_{\xi^\varP}\homega^a_{\dR}) \otimes \sqrt{|k|}\; \theta^1 \ordwedge \cdots \ordwedge \theta^n \otimes \xi_a
\\
+
\sum_a (-1)^{|\homega^a_{\dR}|} \homega^a_{\dR} \otimes (i_\xi \sqrt{|k|}\; \theta^1 \ordwedge \cdots \ordwedge \theta^n ) \otimes \xi_a
=0.
\end{multline*}
Because $\kg$ is unimodular, one has $L^\kg_\xi \sqrt{|k|}\; \theta^1 \ordwedge \cdots \ordwedge \theta^n = 0$, so that $\int_{\algebraic} \homega = \sum_a \homega^a_{\dR} \otimes \xi_a$ is invariant. Looking at each bidegrees for the horizontality condition on $\homega$, one gets $\sum_a (i_{\xi^\varP}\homega^a_{\dR}) \otimes \sqrt{|k|}\; \theta^1 \ordwedge \cdots \ordwedge \theta^n \otimes \xi_a = 0$ so that $\int_{\algebraic} \homega$ is horizontal.

The second point to check is that $\int_{\algebraic}$ coincides on $\kg_\equ$-basic forms with $\int_\inner$. In order to do that, we consider these integrations on a trivialization of $\varP$ given by a local section $s : U \rightarrow \varP$. Then one has the following diagram:
\begin{equation*}
\xymatrix@R=8ex@C=30pt{
  {\Omega^\grast_\tla(\varP,\kg)_{\kg_\equ}} \ar[d]_-{\int_{\algebraic}} \ar[r]^-{s^\ast}
& {\Omega^\grast_\tla(U,\kg)} \ar[d]^-{\int_\inner}
\\
  {(\Omega^{\grast-n}(\varP)\otimes \kg)_{\text{tensorial}}} \ar[r]^-{s^\ast}  
& {\Omega^{\grast-n}(U) \otimes \kg} 
 }
\end{equation*}
The map $s^\ast$ (see for instance \cite{KobaNomi96a}) in the bottom row is the same as the map $s^\ast$ in the top row. For any basic form $\homega = \sum_a (-1)^n \sqrt{|k|}\; \homega^a_{\dR} (\omega_\nabla^{1} - \theta^1) \ordwedge \cdots \ordwedge (\omega_\nabla^{n} - \theta^n) \otimes \xi^a \in \Omega^\grast_\tla(\varP,\kg)_{\kg_\equ}$ one has 
$s^\ast \homega = \sum_a (-1)^n \sqrt{|k|}\; (s^\ast\homega^a_{\dR})\, \lfc^{1} \ordwedge \cdots \ordwedge \lfc^{n}  \otimes \xi^a$ 
because of Lemma~\ref{lem-pullbackvolumeform}, so that
$\int_\inner s^\ast \homega = \sum_a (s^\ast\homega^a_{\dR}) \otimes \xi^a$. 
On the other hand, one has
$\int_{\algebraic} \homega = \sum_a \homega^a_{\dR} \otimes \xi^a$
so that
$s^\ast \int_{\algebraic} \homega = \sum_a (s^\ast\homega^a_{\dR}) \otimes \xi^a$.

This proves the coincidence of the two integrals, because $\int_\inner$ on $\Omega^\grast_\lie(\varP, \kg)$ is completely determined by $\int_\inner$ on the trivializations of forms in $\Omega^\grast_\tla(U,\kg)$.
\end{proof}

Using Prop.~\ref{prop-tripleforinnernondegeneratemetric}, one can define an inner non degenerate metric $\hg$ on $\Gamma_G(T\varP)$ as a triple $(g, h, \nabla)$ where $h$ and $\nabla$ are defined as above and $g$ is an ordinary metric on the base manifold $\varM$. Then the properties of this triple are exactly the ones defining a metric for a non-abelian Kaluza-Klein theory on $\varP$ \cite{Kerner1988fn}. Notice that the geometrical point of view is generally adopted in these theories (geodesics and trajectories of particle) while our point of view here is the one from field theories.

\medskip
Denote by $\caG(\varP)$ the gauge group of $\varP$, of vertical automorphisms of $\varP$. In the following we represent an element $u \in \caG(\varP)$ as a $G$-equivariant map $u : \varP \rightarrow G$, $u(p\cdotaction g) = g^{-1} u(p) g$. It is well known that $\lieL$ is the Lie algebra of $\caG(\varP)$. 

Let $\omega \in \Omega^1_\lie(\varP, \kg)$. For any $u \in \caG(\varP)$, let us introduce $\omega^u(\sfX) = u^{-1} \omega(\sfX) u + u^{-1} (\sfX \cdotaction u)$ for any $\sfX \in \Gamma_G(T\varP)$. Using the $G$-equivariance of $\omega(\sfX)$ and $u$, it is straightforward to verify that $\omega^u(\sfX) \in \Gamma_G(\varP, \kg)$, so that $\omega^u \in \Omega^1_\lie(\varP, \kg)$.

The map $\omega \mapsto \omega^u$ is the action of $\caG(\varP)$ on connections on $\Gamma_G(T\varP)$. This action reduces to the infinitesimal action of $\lieL = \Gamma_G(\varP, \kg)$ defined in \cite{Mass38} and used in Section~\ref{sec-gaugetheories}.

Let $\varE = \varP \times_{\ell} \evF$ be an associated vector bundle for a vector space $\evF$ supporting a representation $\ell$ of $G$. Then $\phi(\sfX)\varphi = \sfX\cdotaction \varphi$ is a representation of $\Gamma_G(T\varP)$ where we look at sections of $\varE$ as $G$-equivariant maps $\varphi: \varP \rightarrow \evF$ such that $\varphi(p\cdotaction g) = \ell(g^{-1}) \varphi(p)$. The covariant derivative $\hnabla^\varE$ associated to a connection $\omega$ on $\Gamma_G(T\varP)$ is given by $\hnabla^\varE_\sfX \varphi = \sfX \cdotaction \varphi + \ell(\omega(\sfX)) \varphi$ where $\ell$ designates also the induced representation of $\kg$ on $\evF$. The gauge group $\caG(\varP)$ acts on $\Gamma(\varE)$ by $\varphi \mapsto \ell(u^{-1}) \varphi$. The covariant derivative associated to $\omega^u$ is $\hnabla^{\varE\, u}_\sfX \varphi  = \ell(u^{-1}) \hnabla^\varE_\sfX \ell(u) \varphi$, which is the usual expression of a gauge transformation on a covariant derivative.

To any $u \in \caG(\varP)$ we can associate the two $1$-forms $u^{-1} \dd u = u^\ast \theta \in \Omega^1(\varP)\otimes \kg$ and $u^{-1} \ds u = \Ad_{u^{-1}} \theta - \theta \in \kg^\dualast \otimes \kg$, so that $u^{-1} \hd_\tla u =  u^{-1}\dd u + u^{-1} \ds u \in \Omega^1_\tla(\varP,\kg)$ makes sense.

Let $\homega \in \Omega^1_\tla(\varP,\kg)$ be the $\kg_\equ$-basic $1$-form corresponding to $\omega \in \Omega^1_\lie(\varP, \kg)$. Then the gauge transformation on $\homega$ takes the form $\homega^u = u^{-1} \homega u + u^{-1} \hd_\tla u$. To verify that $\homega^u$ is $\kg_\equ$-basic requires to take into account all the $G$-equivariances of the various objects in this relation.

The Killing metric $k$ on $\kg$ defines a locally constant Killing inner metric $h$ on $\lieL$. Using a metric on $\varM$ and a background connection $\omegadot$ on $\varP$, one can write an action functional for connections on $\Gamma_G(T\varP)$. This action functional reduces to the ordinary Yang-Mills action functional on ordinary connections, as it is easy to see by imposing $\rke=0$. In the same way, the functional action of matter fields reduces to the usual action functional of minimal coupling with Yang-Mills potentials.

This means that the gauge theories proposed in Section~\ref{sec-gaugetheories}, when specified on Atiyah Lie algebroids, are generalizations of the Yang-Mills gauge theories used in physics. We have noticed in \ref{subsec-physicalremarks} that these theories are of the Yang-Mills-Higgs types, and the similitude of this approach with the one proposed in non-commutative geometry has already been noticed. But there are two main differences we would like to highlight. 

The first one is that in the present approach, the ``generalized'' gauge theories constructed here contains the ordinary gauge theories, so that we can consider the latter as \emph{special cases} in a larger class of theories. In particular, the Atiyah Lie algebroid framework gets in close contact with the ordinary geometry of fibre bundles and their connections, which are at the heart of present day gauge theories. 

The second point is more technical. One of the problems encountered in the non-commutative approach is the fact that the gauge group is extracted from some algebraic structures, for instance as automorphisms of an associative algebra. Here, every (usual) gauge group can be promoted into these new gauge theories, because these gauge groups are related to some principal fibre bundle, which in turn gives rise to an Atiyah Lie algebroid.

\subsection{Derivations on a vector bundle}
\label{subsec-Derivationsonavectorbundle}

Let $\varE$ be a rank $p$ complex vector bundle over the manifold $\varM$. Using any hermitian structure on $\varE$, we suppose that its structure group $H$ is contained in $U(p)$, the group of complex unitary $p\times p$ matrices. Denote by $\End(\varE) = \varE \otimes \varE^\ast$ the fibre bundle of endomorphisms of $\varE$ where $\varE^\ast$ is the dual vector bundle associated to $\varE$. Denote by $\algA(\varE) = \Gamma(\End(\varE))$ the algebra of endomorphisms of $\varE$.

Let $\kD(\varE)$ be the space of first order operators on $\Gamma(\varE)$ whose symbol is the identity. Then $\xymatrix@1@C=15pt{{\algzero} \ar[r] & {\algA(\varE)} \ar[r]^-{\iota} & {\kD(\varE)} \ar[r]^-{\sigma} & {\Gamma(T \varM)} \ar[r] & {\algzero}}$ is the transitive Lie algebroid of derivations of $\varE$ where $\sigma$ is the symbol map \cite{MR585879}.

Denote by $(\Omega^\grast_\lie(\varE, \algA(\varE)), \hd)$ the graded differential algebra of forms on this transitive Lie algebroid with values in its kernel, and denote by $(\Omega^\grast_\lie(\varE), \hd_\lie)$ the graded commutative differential algebra of forms on $\kD(\varE)$ with values in $C^\infty(\varM)$. The natural inclusion $C^\infty(\varM) \rightarrow \algA(\varE)$ induces a morphism of graded differential algebras $\Omega^\grast_\lie(\varE) \hookrightarrow \Omega^\grast_\lie(\varE, \algA(\varE))$.

Let $\{ U_i, \phi_i \}_{i \in I}$ be a system of trivializations of $\varE$ associated to a good cover $\{ U_i \}_{i \in I}$ of $\varM$, where $\phi_i : U_i \times \gC^p \rightarrow \varE_{|U_i}$ are linear isomorphisms. Then this system of trivializations induces a natural system of trivializations of $\End(\varE)$, $\{ U_i, \hphi_i \}_{i \in I}$, such that $\hphi_i : U_i \times M_p(\gC) \rightarrow \End(\varE)_{|U_i}$ and $\hphi_i(x,\gamma) \cdotaction \phi_i(x,v) = \phi_i(x, \gamma \cdotaction v)$ for any $\gamma \in M_p(\gC)$ and any $v \in \gC^p$. Any $s \in \Gamma(\varE)$ (resp. $a \in \algA(\varE)$) is then trivialized by a family of maps $s^i : U_i \rightarrow \gC^p$ (resp. $a^i : U_i \rightarrow M_p(\gC)$). A first order operator $\kX \in \kD(\varE)$ is trivialized as a family of elements $X_i \oplus \gamma^i \in \Gamma(T U_i) \oplus \Gamma(U_i \times M_p(\gC))$ through the relation
\begin{equation*}
(\kX \cdotaction s)(x) = \phi_i \left(x, (X_i \cdotaction s_i)(x) + \gamma^i(x) \cdotaction s_i(x) \right)
\end{equation*}
for any $x \in U_i$ where $\cdotaction$ means either the action of a vector field on vector valued functions or the action of matrices on vectors. Notice that $X_i = X_j = X = \sigma(\kX)$ on $U_{ij} \neq \ensvide$ and $\gamma^i = h_{ij} \gamma^j h_{ij}^{-1} + h_{ij} \dd h_{ij}^{-1} (X)$ where $h_{ij} : U_{ij} \rightarrow H \subset U(p)$ are the transition functions of $\varE$ such that $s^i(x) = h_{ij}(x) s^j(x)$ for any $x \in U_{ij}$.
The system of trivializations considered for $\kD(\varE)$ is thus defined by $\nabla^{0,i}_X = X$ and $\Psi_i(a^i) = a^i$ for any $a^i : U_i \rightarrow M_p(\gC)$, and one has $\alpha_j^i(\gamma) = h_{ij} \gamma h_{ij}^{-1}$ and $\chi_{ij}(X) = h_{ij} \dd h_{ij}^{-1}(X)$. The local description of $\Omega^\grast_\lie(\varE, \algA(\varE))$ is given by the differential calculi $\Omega^\grast(U_i) \otimes \exter^\grast M_p^\ast \otimes M_p$.

The Lie algebra on which this Lie algebroid is modelled is $\kg = M_p(\gC) = M_p$ with the commutator as Lie bracket and $n= p^2$. This Lie algebra decomposes as $\kg = \gC \bbbone_p \oplus \ksl_n$ where $\bbbone_p$ is the unit matrix in $M_p$ and $\ksl_p$ is the Lie algebra of traceless matrices in $M_p$. In any trivialization, one can decompose $X \oplus \gamma^i = \tla(U_i, M_p)$ as $X \oplus \left( \frac{1}{p} \lambda^i \bbbone_p \oplus \gamma_0^i \right)$ where $\lambda^i = \tr(\gamma^i)$ and $\gamma_0^i = \gamma^i - \frac{1}{p} \lambda^i \bbbone_p : U_i \rightarrow \ksl_p$. Then the family $X \oplus \lambda^i$ associated to a family $X_i \oplus \gamma^i$ of trivializations of an element $\kX \in \kD(\varE)$ defines a global element in the transitive Lie algebroid
\begin{equation*}
\xymatrix@1{{\algzero} \ar[r] & {C^\infty(\varM)} \ar[r]^-{\iota} & {\kD(\det(\varE))} \ar[r]^-{\sigma} & {\Gamma(T \varM)} \ar[r] & {\algzero}}
\end{equation*}
where $\det(\varE) = \exter^p \varE$ is the determinant line bundle associated to $\varE$. This map is the natural representation of $\kD(\varE)$ on $\det(\varE)$ given by $\kX \mapsto \exter^p \kX$ where
\begin{equation*}
\left( \exter^p \kX \right)(e_1 \ordwedge \cdots \ordwedge e_p) = \sum_{k=1}^p e_1 \ordwedge \cdots \ordwedge \kX(e_k) \ordwedge \cdots \ordwedge e_p.
\end{equation*}
The induced map $\algA(\varE) \rightarrow C^\infty(\varM)$ is the globally defined trace $\tr$, which is a Lie morphism. This representation gives rise to a natural morphism of graded commutative differential algebras $\Omega^\grast_\lie(\det(\varE)) \rightarrow \Omega^\grast_\lie(\varE)$.

The inner orientability of $\kD(\varE)$ corresponds to the orientability of the (vector) bundle $\End(\varE)$. Because $U(p)$ is unimodular, $\End(\varE)$ is always orientable. The trace map given before defines a non degenerate inner metric $h$ on $\algA(\varE)$ given by $h(a,b) = \tr(ab)$. In any local trivialization, the inner metric $h$ is represented by a constant matrix. Notice that the unimodularity of the (real) Lie algebra $\ku_p$ of $U(p)$ implies the unimodularity of the (complex) Lie algebra $\kg = M_p$.

The inner integration $\int_\inner : \Omega^\grast_\lie(\varE, \algA(\varE)) \rightarrow \Omega^{\grast-n}(\varM, \End(\varE))$ defined by the inner metric $h$ can be composed with the trace map in order to define
\begin{equation}
\label{eq-definnerintegrationtrace}
\int^{\tr}_\inner = \tr \circ \int_\inner : \Omega^\grast_\lie(\varE, \algA(\varE)) \rightarrow \Omega^{\grast-n}(\varM)
\end{equation}

\begin{proposition}
For any $\omega \in \Omega^\grast_\lie(\varE, \algA(\varE))$ one has
\begin{equation*}
\int^{\tr}_\inner \hd \omega = \dd \int^{\tr}_\inner \omega
\end{equation*}
\end{proposition}

In some extent, this Proposition generalizes the result obtained in Theorem~\ref{thm-atiyah-commutationofdifferentials} for forms which are not $C^\infty(\varM)$-valued. In doing so, it completes the first row of Theorem~\ref{thm-relationsdeRhamTLAAtiyah} using the trace in order to end in the space of forms on $\varM$.

\begin{proof}
The differential $\hd$ is locally the sum of three parts on $\Omega^\grast(U_i) \otimes \exter^\grast M_p^\ast \otimes M_p$: the de~Rham differential on $\Omega^\grast(U_i)$, the Chevalley-Eilenberg differential on $\exter^\grast M_p^\ast$ and the adjoint action on $M_p$. Using similar arguments as the ones used in the proof of Theorem~\ref{thm-atiyah-commutationofdifferentials} and the fact that the trace kills the adjoint action on $M_p$, one gets the result.
\end{proof}

Let $\der(\algA(\varE))$ be the Lie algebra and $C^\infty(\varM)$-module of derivations of the associative algebra $\algA(\varE)$. In \cite{Mass14}, a natural surjection $\kD(\varE) \rightarrow \der(\algA(\varE))$ was proposed: it associates to any $\kX \in \kD(\varE)$ the derivation $a \mapsto [\kX, a]$ for any $a \in \algA(\varE)$ where the commutator takes place in the space of operators on $\Gamma(\varE)$. Locally this corresponds to $X \oplus \gamma \mapsto X \oplus \ad_\gamma$.

Now, if the structure group $H$ of $\varE$ can be reduced such that $H \subset SU(p)$, then there is a natural injection $\der(\algA(\varE)) \rightarrow \kD(\varE)$ of Lie algebroids defined locally by $X \oplus \ad_{\gamma^i} \mapsto X \oplus \gamma^i$ for any traceless $\gamma^i : U_i \rightarrow \ksl_p$. One then has a splitting of Lie algebroids $\kD(\varE) \simeq \der(\algA(\varE)) \oplus C^\infty(\varM)$. The inclusion of $\der(\algA(\varE))$ into $\kD(\varE)$ induces a natural morphism of graded differential algebras $\Omega^\grast_\lie(\varE,\algA(\varE)) \rightarrow \Omega^\grast_\der(\algA(\varE))$ where $\Omega^\grast_\der(\algA(\varE))$ is the derivation based differential calculus associated to $\algA(\varE)$ (see \cite{Mass30,Mass38} for details). This morphism connects together the integration $\int^{\tr}_\inner$ defined in \eqref{eq-definnerintegrationtrace} and the equivalent integration defined in \cite{Mass15} on the non-commutative geometry of the algebra $\algA(\varE)$.

\section*{References}
\bibliography{Gauge-theories-Lie-algebroids}

\end{document}